\documentclass[11pt,letter]{article}
\usepackage{subfigure}
\usepackage{url}
\makeatletter
\newcommand{\singlespacing}{\let\CS=
\@currsize\renewcommand{\baselinestretch}{1}\tiny\CS}
\newcommand{\singlespacingplus}{\let\CS=
\@currsize\renewcommand{\baselinestretch}{1.25}\tiny\CS}
\newcommand{\doublespacing}{\let\CS=
\@currsize\renewcommand{\baselinestretch}{1.75}\tiny\CS}
\newcommand{\extradoublespacing}{\let\CS=
\@currsize\renewcommand{\baselinestretch}{1.9}\tiny\CS}
\newcommand{\nicenicespacing}{\let\CS=
\@currsize\renewcommand{\baselinestretch}{1.9}\tiny\CS}
\newcommand{\draftspacing}{\let\CS=
\@currsize\renewcommand{\baselinestretch}{2.0}\tiny\CS}
\newcommand{\hugedraftspacing}{\let\CS=
\@currsize\renewcommand{\baselinestretch}{2.4}\tiny\CS}
\newcommand{\niceonespacing}{\let\CS=\@currsize\renewcommand{\baselinestretch}{1.1}\tiny\CS}
\newcommand{\nicetwospacing}{\let\CS=\@currsize\renewcommand{\baselinestretch}{1.2}\tiny\CS}
\newcommand{\nicethreespacing}{\let\CS=\@currsize\renewcommand{\baselinestretch}{1.3}\tiny\CS}
\newcommand{\singlespacingplusplus}{\let\CS=\@currsize\renewcommand{\baselinestretch}{1.35}\tiny\CS}
\newcommand{\nicefourspacing}{\let\CS=\@currsize\renewcommand{\baselinestretch}{1.4}\tiny\CS}
\newcommand{\nicefivespacing}{\let\CS=\@currsize\renewcommand{\baselinestretch}{1.5}\tiny\CS}
\newcommand{\nicesixspacing}{\let\CS=\@currsize\renewcommand{\baselinestretch}{1.6}\tiny\CS}
\newcommand{\nicesevenspacing}{\let\CS=\@currsize\renewcommand{\baselinestretch}{1.7}\tiny\CS}
\newcommand{\niceeightspacing}{\let\CS=\@currsize\renewcommand{\baselinestretch}{1.8}\tiny\CS}
\newcommand{\niceninespacing}{\let\CS=\@currsize\renewcommand{\baselinestretch}{1.9}\tiny\CS}
\makeatother%

\makeatletter
\def\@cite#1#2{[#1\if@tempswa , #2\fi]}
\makeatother
\makeatletter
\def\@citex[#1]#2{\if@filesw\immediate\write\@auxout{\string\citation{#2}}\fi
  \def\@citea{}\@cite{\@for\@citeb:=#2\do
    {\@citea\def\@citea{,\linebreak[0]}\@ifundefined
       {b@\@citeb}{{\bf ?}\@warning
       {Citation `\@citeb' on page \thepage \space undefined}}%
\hbox{\csname b@\@citeb\endcsname}}}{#1}}
\makeatother
\makeatletter
\def\@cite#1#2{[#1\if@tempswa , #2\fi]}
\def\@citex[#1]#2{\if@filesw\immediate\write\@auxout{\string\citation{#2}}\fi
  \def\@citea{}\@cite{\@for\@citeb:=#2\do
    {\@citea\def\@citea{,\kern1pt\linebreak[0]}\@ifundefined
       {b@\@citeb}{{\bf ?}\@warning
       {Citation `\@citeb' on page \thepage \space undefined}}%
\hbox{\csname b@\@citeb\endcsname}}}{#1}}
\makeatother

\newcount\hour  \newcount\minutes  \hour=\time  \divide\hour by 60
\minutes=\hour  \multiply\minutes by -60  \advance\minutes by \time
\def\mmmddyyyy{\ifcase\month\or Jan\or Feb\or Mar\or Apr\or May\or Jun\or Jul\or
  Aug\or Sep\or Oct\or Nov\or Dec\fi \space\number\day, \number\year}
\def\hhmm{\ifnum\hour<10 0\fi\number\hour :%
  \ifnum\minutes<10 0\fi\number\minutes}

\usepackage{amsmath}
\usepackage{amssymb}
\usepackage{amsthm}

\usepackage{epsfig}
\usepackage{multicol}

\usepackage{algorithmic}
\usepackage{algorithm}
\newtheorem{prop}{Proposition}[section]
\newtheorem{theorem}[prop]{Theorem}

\newtheorem{cor}[prop]{Corollary}
\newtheorem{lemma}[prop]{Lemma}

\newtheorem*{conventionA}{Convention A}

\newtheorem{claim}[prop]{Claim}

\newtheorem{definition}[prop]{Definition}

\def\calC{\mathcal{C}}

\def\calS{\mathcal{S}}

\newcommand{\revnot}[1]{\overleftarrow{#1}}

\newcommand{\comment}[1]{}

\newcommand{\electionsystem}{{\cal{E}}}

\newcommand{\naturals}{\mathbb{N}}

\newcommand{\win}{{\it win}}
\newcommand{\ac}{{\rm AC}}
\newcommand{\dc}{{\rm DC}}
\newcommand{\av}{{\rm AV}}
\newcommand{\dv}{{\rm DV}}
\newcommand{\bv}{{\rm BV}}

\newcommand{\p}{{\rm P}}
\newcommand{\np}{{\rm NP}}

\newcommand{\xthreec}{{\rm X3C}}

\newcommand{\calR}{{\cal R}}

\newcommand{\score}{{\ensuremath{\mathrm{score}}}}
\newcommand{\names}{{\ensuremath{\mathrm{names}}}}
\newcommand{\scd}{{\ensuremath{\mathrm{sc}}}}
\newcommand{\df}{{\ensuremath{\mathrm{df}}}}
\newcommand{\copeland}{{\ensuremath{\mathrm{Copeland}}}}
\newcommand{\origllull}{{\ensuremath{\mathrm{OriginalLlull}}}}

\newcommand{\thetatwo}{{{\Theta_2^\mathrm{p}}}}

\begin{document}
\title{Multimode Control Attacks on Elections\thanks{Also appears as URCS-TR-2010-960.}}
\author{Piotr Faliszewski \\
Dept.~of Computer Science\\
        AGH University of Science \\
        and Technology,  Krak\'ow \\
        Poland
\and
        Edith Hemaspaandra \\
Dept.~of Computer Science\\
        Rochester Institute of Technology \\
        Rochester, NY 14623\\ USA
\and
        Lane A. Hemaspaandra \\
Dept.~of Computer Science\\
        University of Rochester \\
        Rochester, NY 14627 \\ USA
}
\date{July 11, 2010}

\maketitle

\begin{abstract}
  In 1992, Bartholdi, Tovey, and Trick
  opened the study of control attacks on elections---attempts to
  improve the election outcome by such actions as adding/deleting
  candidates or voters.  That work has led to many results on how
  algorithms can be used to find attacks on elections and how
  complexity-theoretic hardness results can be used as shields against
  attacks. However, all the work in this line has assumed that the
  attacker employs just a single type of attack.  In this paper, we
  model and study the case in which the attacker launches a
  multipronged (i.e., multimode) attack.  We do so to more
  realistically capture the richness of real-life settings. For
  example, an attacker might simultaneously try to suppress some
  voters, attract new voters into the election, and introduce a
  spoiler candidate. Our model provides a unified framework for such
  varied attacks, and by constructing polynomial-time multiprong
  attack algorithms we prove that for various election systems even
  such concerted, flexible attacks can be perfectly planned in
  deterministic polynomial time.
\end{abstract}

\section{Introduction}\label{sec:intro}

Elections are a central model for collective decision-making: Actors'
(voters') preferences among alternatives (candidates) are input to the
election rule and a winner (or winners in the case of ties) is
declared by the rule.  Bartholdi, 
Orlin, Tovey, and
Trick
initiated a line of research whose goal is to protect elections from
various attacking actions intended to skew the election's 
results.  Bartholdi, Orlin,
Tovey, and Trick's strategy for achieving this goal was to show that
for various election systems and attacking actions, even seeing
whether for a given set of votes such an attack is possible is
$\np$-complete.  Their papers~\cite{bar-tov-tri:j:manipulating,bar-oli:j:polsci:strategic-voting,bar-tov-tri:j:control}
consider actions such as voter manipulation (i.e., situations where a
voter misrepresents his or her vote to obtain some goal) and various
types of election control (i.e., situations where the attacker is
capable of modifying the structure of an election, e.g., by adding or
deleting either voters or candidates). Since then, many researchers
have extended Bartholdi, Orlin, Tovey, and Trick's work by providing new
models, new results, and new perspectives.  But to the best 
of our knowledge, until now no one 
has considered the situation in which an attacker combines multiple
standard attack types into a single attack---let us call that a
\emph{multipronged} (or \emph{multimode}) \emph{attack}.

Studying multipronged control is a step in the direction of more
realistically modeling real-life scenarios. Certainly, in real-life
settings an attacker would not voluntarily limit himself or herself to
a single type of attack but rather would use all available means of
reaching his or her goal. For example, an attacker interested in some
candidate $p$ winning might, at the same time, intimidate $p$'s most
dangerous competitors so that they would withdraw from the election,
and encourage voters who support $p$ to show up to vote. In this paper
we study the complexity of such multipronged control
attacks.\footnote{In fact, our framework of multiprong control
  includes the 
unpriced bribery 
of 
\cite{fal-hem-hem:j:bribery}
  and can be extended to include manipulation.}

Given a type of multiprong control, we seek to analyze its complexity.
In particular, we try to show either that one can 
compute in polynomial time an optimal attack of that control type,
or that even recognizing the existence of an attack is $\np$-hard. 
It
is particularly interesting to ask about the complexity of a
multipronged attack whose components each have efficient algorithms.
We are interested in whether such a combined attack (a)~becomes
computationally hard, or (b)~still has a polynomial-time algorithm.
Regarding the (a) case, we give an example of a natural election
system that displays this behavior.
Our paper's core work studies the (b) case and shows that
even attacks having multiple 
prongs can in many cases be planned with
perfect efficiency.  Such results yield as immediate consequences all
the individual efficient attack algorithms for each prong, and as such
allow a more compact presentation of results and more compact proofs.
But they go beyond that: They show that the interactions between the
prongs can be managed without such cost as to move beyond polynomial
time.

\smallskip\noindent\textbf{Related work.}  Since the seminal 
paper of
Bartholdi, Tovey, and Trick~\cite{bar-tov-tri:j:control}, 
much research
has been dedicated to studying the complexity of control in elections.
Bartholdi, Tovey, and Trick~\cite{bar-tov-tri:j:control}
considered
constructive control only, i.e., scenarios where the goal of the
attacker is to ensure some candidate's victory.  Hemaspaandra,
Hemaspaandra, and Rothe~\cite{hem-hem-rot:j:destructive-control}
extended their work to the destructive case, i.e., scenarios in which
the goal is to prevent someone from winning. 

A central but elusive goal of
control research is finding a natural election system (with a
polynomial-time winner algorithm) that is resistant to all the
standard types of control, i.e., for which all the types of control
are $\np$-hard.  Hemaspaandra, Hemaspaandra, and
Rothe~\cite{hem-hem-rot:j:hybrid} showed that there exist highly
resistant artificial election systems. Faliszewski et
al.~\cite{fal-hem-hem-rot:j:llull} then showed that 
the natural system known as Copeland voting is not
too far from the goal mentioned above.  And Erd\'{e}lyi, Rothe, and
Nowak~\cite{erd-now-rot:j:sp-av} then showed a system with even more
resistances than Copeland, but in a slightly nonstandard voter model
(see~\cite{bau-erd-hem-hem-rot:btoappear:computational-apects-of-approval-voting}
for discussion 
and~\cite{erd-pir-rot:t:fallback,erd-pir-rot:t:bucklin,men:t:range-voting} 
for some related follow-up work).

Recently, researchers also started focusing on the parameterized
complexity of control in elections. Faliszewski et
al.~\cite{fal-hem-hem-rot:j:llull} provided several fixed-parameter
tractability results.  Betzler and
Uhlmann~\cite{bet-uhl:j:parameterized-complecity-candidate-control} and
Liu et al.~\cite{fen-liu-lua-zhu:j:parameterized-control} showed
so-called 
W[1]- and W[2]-hardness results for control under various
voting rules. In response to the conference version of
the present paper~\cite{fal-hem-hem:c:multimode}, 
Liu and Zhu conducted a parameterized-complexity study
of control in maximin elections~\cite{liu-zhu:j:maximin}.

Going in a somewhat different direction, Meir et
al.~\cite{mei-pro-ros-zoh:j:multiwinner} bridged the notions of
constructive and destructive control by considering utility
functions, and in this model obtained control results for multiwinner
elections. In multiwinner elections the goal is to elect a whole group
of people (consider, e.g., parliamentary elections) rather than just a
single person. 
Elkind, Faliszewski, and Slinko~\cite{elk-fal-sli:c:cloning} and
Chevaleyre et al.~\cite{che-mau-mon-lan:c:possible-winners-adding} 
considered two types of problems related to control by adding candidates for 
the case where it is not known how the voters would rank
the added candidates.

Faliszewski et al.~\cite{fal-hem-hem-rot:c:single-peaked-preferences}
 and Brandt et al.~\cite{bra-bri-hem-hem:c:sp2} 
have studied control (and manipulation and bribery) in 
so-called single-peaked domains, a model of overall electorate 
behavior from political science.

There is a growing body of work on manipulation that regards frequency
of (non)hardness of election problems (see,
e.g.,~\cite{con-san:c:nonexistence,fri-kal-nis:c:quantiative-gib-sat,dob-pro:c:two-voters,xia-con:c:frequently-manipulable,xia-con:c:generalized-scoring,wal:c:where-hard-veto}).
This work studies whether a given $\np$-hard election problem (to date
only manipulation/winner problems have been studied, not control 
problems) can be often
solved in practice (assuming some distribution of votes). Such results
are of course very relevant when one's goal is to protect elections
from manipulative actions.  However, in this paper we typically take
the role of an attacker and design control algorithms that are fast on
\emph{all} instances. 

Faliszewski et al.~\cite{fal-hem-hem-rot:b:richer,fal-hem-hem:jtoappear:cacm-survey}
provide an overview of some complexity-of-election issues.

\smallskip\noindent\textbf{Organization.}  In
Section~\ref{sec:elections} we present the standard model of elections
and describe relevant voting systems. In Section~\ref{sec:control} we
introduce multiprong control, provide initial results, and show how
existing immunity, vulnerability, and resistance results interact with
this model.  In Section~\ref{sec:maximin} we provide a complexity
analysis of voter and 
candidate control in maximin elections, showing how
multiprong control is useful in doing so.  
We also show
that maximin has an interesting relation to Dodgson elections:
No candidate whose Dodgson score is more than $m^2$ times
that of the Dodgson winner(s) can be a maximin winner.
In Section~\ref{sec:fpt} we
consider fixed-parameter complexity of multiprong control, using
as our parameter the number of candidates.  
Section~\ref{sec:conclusion} provides
conclusions and open problems.

\section{Preliminaries}
\label{sec:elections}

\textbf{Elections.}  An election is a pair $(C,V)$, where $C = \{c_1,
\ldots, c_m\}$ is the set of candidates and $V = (v_1, \ldots, v_n)$
is a collection of voters. Each voter $v_i$ is represented by his or
her preference list.\footnote{We also assume that each voter has a
  unique name.  However, all the election systems we consider
  here---except for the election system of
  Theorem~\ref{thm:origllull}---are anonymous and thus disregard
  voter names and the order of the votes.}
For example, if we have three candidates, $c_1$,
$c_2$, and $c_3$, a voter who likes $c_1$ most, then $c_2$, and then
$c_3$ would have preference list $c_1 > c_2 >
c_3$.\footnote{Preference lists are also called preference orders, and
  in this paper we will use these two terms interchangeably.}  Given
an election $E = (C,V)$, by $N_E(c_i,c_j)$, where $c_i,c_j \in C$ and
$i \neq j$, we denote the number of voters in $V$ who prefer $c_i$ to
$c_j$.  We adopt the following convention for specifying preference
lists.
\begin{conventionA}
\label{conv:A}
Listing some set $D$ of candidates as an item in a preference list
means listing all the members of this set in some fixed, arbitrary
order, and listing $\revnot{D}$ means listing all the members of $D$,
but in the reverse order. 
\end{conventionA}

An election system is 
a mapping that given an election $(C,V)$
outputs a set $W$, satisfying $W \subseteq C$, called
the winners of the election.

We focus on the following five voting systems: plurality, Copeland,
maximin, approval, and Condorcet. (However, in
Sections~\ref{sec:dodgson} and~\ref{sec:fpt} we take a detour through
some other systems.)  Each of plurality, Copeland, maximin, and
approval assigns points to candidates and elects those that receive
the most points. Let $E = (C,V)$ be an election, where $C = \{c_1,
\ldots,c_m\}$ and $V = (v_1, \ldots, v_n)$.  In plurality, each
candidate receives a single point for each voter who ranks him or her
first. In maximin, the score of a candidate $c_i$ in $E$ is defined as
$\min_{c_j \in C - \{c_i\}}N_E(c_i,c_j)$. For each rational $\alpha$,
$0 \leq \alpha \leq 1$, in Copeland$^\alpha$ candidate $c_i$ receives
$1$ point for each candidate $c_j$, $j \neq i$, such that
$N_E(c_i,c_j) > N_E(c_j,c_i)$ and $\alpha$ points for each candidate
$c_j$, $j\neq i$, such that $N_E(c_i,c_j) = N_E(c_j,c_i)$. That is,
the parameter $\alpha$ describes the value of ties in head-to-head
majority contests. In approval, instead of preference lists each voter's
ballot is 
a 0-1 vector, where each entry denotes whether the voter approves
of the corresponding candidate (gives the corresponding candidate a
point). For example, vector $(1,0,0,1)$ means that the voter approves
of the first and fourth candidates, but not the second and 
third.
We use $\score_E(c_i)$ to denote the score of candidate $c_i$ in
election $E$ (the particular election system used will always be clear
from context).

A candidate $c$ is a Condorcet winner of an election
$E = (C,V)$ if for each other candidate $c' \in C$ it holds that
$N_E(c,c') > N_E(c',c)$.  Clearly, each election has at most one
Condorcet winner.  (Not every election has a Condorcet winner.
However, as our notion of an election allows outcomes in which no one
wins, electing the Condorcet winner when there is one and otherwise
having no winner is a quite legal election system.)

\smallskip
\noindent
\textbf{Computational complexity.}  We use standard notions of
complexity theory, as presented, e.g., in the textbook of
Papadimitriou~\cite{pap:b:complexity}. We assume that the reader is
familiar with the complexity classes $\p$ and $\np$, polynomial-time many-one
reductions, and the notions of $\np$-hardness and $\np$-completeness.
$\naturals$ will denote
$\{0,1,2,\ldots\}$.

Most of the $\np$-hardness proofs in this paper follow by a reduction
from the well-known $\np$-complete problem \emph{exact cover by
  3-sets}, known for short
as $\xthreec$ (see, e.g.,~\cite{gar-joh:b:int}).  In
$\xthreec$ we are given a pair $(B,\calS)$, where $B = \{b_1, \ldots,
b_{3k}\}$ is a set of $3k$ elements and $\calS = \{S_1, \ldots, S_n\}$
is a set of $3$-subsets of $B$, and we ask whether there is a subset $S'$
of exactly $k$ elements of $\calS$ such that their union is exactly
$B$.  We call such a set $S'$ an exact cover of~$B$.

In Section~\ref{sec:fpt}, we consider the fixed-parameter complexity of
multiprong control. The idea of fixed-parameter complexity is to
measure the complexity of a given decision problem with respect to
both the instance size (as in the standard complexity theory) and some
parameter of the input (in our case, the number of candidates
involved).  For a problem to be said to 
be fixed-parameter tractable, i.e., to belong to the 
complexity class FPT, we as is standard require that the problem
can be solved by an algorithm running in 
time $f(j)n^{O(1)}$, where $n$ is the size of
the encoding of the given instance, $j$ is the value of the parameter
for this instance, and $f$ is some function. Note that $f$
does not have to be polynomially bounded or even 
computable.  However, in all FPT claims in this paper, $f$ is 
a computable function.  That is, our algorithms actually achieve 
so-called strongly uniform fixed-parameter
tractability.
We point readers interested in parameterized
complexity to, for example, 
Niedermeier's book~\cite{nie:b:invitation-fpt}.

\section{Control and Multiprong Control}
\label{sec:control}

In this section we introduce multiprong control, that is, control
types that combine several standard types of control. We first provide
the definition, then proceed to analyzing general properties of
multiprong control, then consider multiprong control for
election systems for which the complexity of single-prong control has
already been established, and finally give an example of an
election system for which multiprong control becomes harder than
any of its constituent prongs (assuming $\p\neq\np$).

\subsection{The Definition}

We consider combinations of control by
adding/deleting candidates/voters\footnote{Other control types,
  defined by Bartholdi, Tovey, and Trick~\cite{bar-tov-tri:j:control}
  and refined by Hemaspaandra, Hemaspaandra, and
  Rothe~\cite{hem-hem-rot:j:destructive-control}, regard various types
  of partitioning candidates and voters.}  and by bribing
voters. Traditionally, bribery has not been considered a type of
control but it fits the model very naturally and strengthens our
results.

In discussing control problems, we must be very clear about
whether the goal of the attacker is to make his or her preferred
candidate \emph{the only winner}, or is to make his or her
preferred candidate \emph{a winner}.  To be clear on this,
we as is standard will use the term ``unique-winner model'' for the
model in which the goal is to make one's preferred candidate the one
and only winner, and we will use the term ``nonunique-winner model''
for the approach in which the goal is to make one's preferred
candidate be a winner.  (Note that 
if exactly one person wins, he or she most certainly is considered 
to have satisfied the control action in the nonunique-winner 
model.  The ``nonunique'' in the model name merely means 
we are not \emph{requiring} that winners be unique.)

The destructive cases of each of these are, 
in the nonunique-winner model, blocking one's despised
candidate from being a unique winner,\footnote{We will often use the
  phrase ``a unique winner,'' as we just did.
  The reason we write ``a unique winner'' rather than ``the unique
  winner'' is to avoid the impression that the election necessarily
  has some (unique) winner.}  and in the unique-winner model, blocking
one's despised candidate from being a winner.  We take the
unique-winner model as the default in this paper, as is the most
common model in studies of control.

\begin{definition}\label{def:multiprong}
  Let $\electionsystem$ be an election system. In the
  unique-winner,\footnote{One can easily adapt the definition to the
    nonunique-winner model.} constructive  $\electionsystem$-AC+DC+AV+DV+BV
  control problem we are given:
  \begin{enumerate}
  \item[(a)] two disjoint sets of candidates, $C$ and $A$,
  \item[(b)] two disjoint collections of voters, $V$ and $W$,
    containing voters with preference lists over $C \cup A$,
  \item[(c)] a preferred candidate $p \in C$, and
  \item[(d)] five nonnegative integers, $k_{\ac}$, $k_{\dc}$,
    $k_{\av}$, $k_{\dv}$, and $k_\bv$.
  \end{enumerate}
  We ask whether it is possible to find two sets, $A' \subseteq A$ and
  $C' \subset C$, and two subcollections of voters, $V' \subseteq V$
  and $W' \subseteq W$, such that:
  \begin{enumerate}
  \item[(e)] it is possible to ensure that $p$ is a unique winner
    of $\electionsystem$ election $((C-C')\cup A', (V-V') \cup W')$ via
    changing preference orders of (i.e., bribing) at most $k_\bv$ voters in $(V-V')
    \cup W'$,
  \item[(f)] $p \notin C'$, and
  \item[(g)] $\|A'\| \leq k_\ac$, $\|C'\| \leq k_\dc$, $\|W'\| \leq
    k_\av$, and $\|V'\| \leq k_\dv$.
  \end{enumerate}
  In the unique-winner, destructive variant of the problem, we replace
  item (e) above with: ``it is possible to ensure that $p$ is
  \emph{not} a unique winner of $\electionsystem$ election
  $((C-C')\cup A', (V-V') \cup W')$ via changing preference orders of
  at most $k_\bv$ voters in $(V-V') \cup W'$.''  (In addition, in the
  destructive variant we refer to $p$ as ``the despised candidate''
  rather than as ``the preferred candidate,'' and we often denote him
  or her by $d$.)
\end{definition}

The phrase AC+DC+AV+DV+BV in the problem name corresponds to four of
the standard types of control: adding candidates (AC), deleting
candidates (DC), adding voters (AV), deleting voters (DV), and to
(unpriced) bribery (BV); we will refer to these five types of control
as the \emph{basic} types of control. We again remind the reader
that traditionally bribery is not a type of control but we will call
it a basic type of control for the sake of uniformity and throughout
the rest of the paper we will consider it as such.

Instead of considering all of AC, DC, AV, DV, and BV, we often 
are interested in some subset of them and so we consider special
cases of the AC+DC+AV+DV+BV problem. For example, we write DC+AV 
to refer to a variant of the AC+DC+AV+DV+BV problem
where only deleting candidates and adding voters is allowed.
As part of our model we assume that in each such variant, only the 
parameters relevant to the prongs are part of the input.
So, for example, DC+AV would have $k_\dc$, $k_\av$, $C$,
$V$, $W$, and $p$ as the (only) parts of its input.  And the 
``missing'' parts (e.g., for DC+AV, the missing 
parts are $A$, $k_\ac$, $k_\dv$, and $k_\bv$)
are treated in the obvious way 
in evaluating the formulas 
in Definition~\ref{def:multiprong}, namely,
missing sets are treated as $\emptyset$ and missing constants 
are treated as~$0$.  
If we name only a single type of
control, we in effect degenerate to 
one of the standard control problems. We
for historical reasons consider also a special case of the AC control
type, denoted AC$_{\rm u}$ (and called control by adding an unlimited
number of candidates), where there is no limit on the number of
candidates to add, i.e., $k_\ac = \|A\|$.

There is at least one more way in which we could define multiprong
control. The model in the above definition can be called the
\emph{separate-resource model}, as the extent to which we can use each
basic type of control is bounded separately. In the
\emph{shared-resource model} one pool of action allowances must be
allocated among the allowed control types (so in the definition above
we would replace $k_\ac, k_\dc, k_\av$, and $k_\dv$ with a single
value, $k$, and require that $\|C'\|+\|D'\|+\|V'\|+\|W'\| + 
\mathit{the\hbox{-}number\hbox{-}of\hbox{-}bribed\hbox{-}voters} \leq k$).
Although one could make various
arguments about which model is more appropriate,
their computational complexity is related.

\begin{theorem}
  If there is a polynomial-time algorithm for a given variant of
  multiprong control in the separate-resource model then there is one
  for the shared-resource model as well.
\end{theorem}
\begin{proof}
  Let $\electionsystem$ be an election system.  We will describe the
  idea of our proof on the example of the constructive
  $\electionsystem$-AC+AV problem. The idea easily generalizes to any
  other set of allowed control actions (complexity-theory savvy
  readers will quickly see that we, in essence, give a disjunctive
  truth-table reduction).

  We are given an instance $I$ of the constructive
  $\electionsystem$-AC+AV problem in the shared-resource model, where
  $k$ is the limit on the sum of the number of candidates and voters
  that we may add. Given a polynomial-time algorithm for the
  separate-resource variant of the problem, we solve $I$ using the
  following method.  (If $k > \|A\| + \|W\|$ then set $k = \|A\| +
  \|W\|$.)  We form a sequence $I_0, \ldots, I_{k}$ of instances of
  the separate-resource variant of the problem, where each $I_\ell$,
  $0 \leq \ell \leq k$, is identical to $I$, except that we are allowed
  to add at most $\ell$ candidates and at most $k-\ell$ voters. We
  accept if at least one of $I_\ell$ is a ``yes''-instance of the
  separate-resource, constructive $\electionsystem$-AC+AV
  problem. Clearly, this algorithm is correct and runs in polynomial
  time.
\end{proof}

It would be interesting to consider a variant of the shared-resource
model where various actions come at different costs (e.g., adding some
candidate $c'$ might be much more expensive---or difficult---than
adding some other candidate $c''$). This approach would be close in
spirit to priced bribery of~\cite{fal-hem-hem:j:bribery}. Analysis of
such priced control is beyond the scope of the current paper.

\subsection{Susceptibility, Immunity, Vulnerability,
  and Resistance}
As is standard in the election-control
(and election-bribery) literature, we consider vulnerability, immunity,
susceptibility, and resistance to control. Let $\electionsystem$ be an
election system and let $\calC$ be a type of control. We say that
$\electionsystem$ is susceptible to constructive (destructive) $\calC$
control if there is a scenario in which effectuating $\calC$ makes
someone become (stop being) a unique winner of some
$\electionsystem$ election $E$. $\electionsystem$ 
is immune to
constructive (destructive) $\calC$ control if $\electionsystem$ 
is not susceptible
to 
constructive (destructive) $\calC$ control.
We say that $\electionsystem$ is vulnerable to 
constructive
(destructive)
$\calC$ control if 
$\electionsystem$ 
is susceptible to 
constructive
(destructive)
$\calC$ and
there is a polynomial-time algorithm that decides the constructive
(destructive) $\electionsystem$-$\calC$ problem.  
Actually, this paper's
vulnerability
algorithms/proofs 
will each go further and will in polynomial time
produce, or will make implicitly clear
how to produce, the successful control
action.  So we in each case are even achieving 
the so-called certifiable vulnerability of 
Hemaspaandra, Hemaspaandra, 
and Rothe~\cite{hem-hem-rot:j:destructive-control}.
  $\electionsystem$ is resistant to constructive
(destructive) $\calC$ control if 
$\electionsystem$ 
is susceptible to 
(destructive) $\calC$ control 
and the
constructive (destructive) $\electionsystem$-$\calC$ problem is
$\np$-hard.

The next three theorems describe how multiprong control problems can
inherit susceptibility, immunity, vulnerability, and resistance from the basic
control types that they are built from.

\begin{theorem}\label{thm:propagate}
  Let $\electionsystem$ be an election system and let $C_1+\cdots+C_k$
  be a variant of multiprong control (so $1 \leq k \leq 5$ and each
  $C_i$ is a basic type of control).  $\electionsystem$ is susceptible
  to constructive (destructive) $C_1+\cdots+C_k$ control if and only
  if $\electionsystem$ is susceptible to at least one of constructive
  (destructive) $C_1, \ldots, C_k$ control.
\end{theorem}
\begin{proof}
  The ``if'' direction is trivial: The attacker can always choose to
  use only the type of control to which $\electionsystem$ is
  susceptible.  As to the ``only if'' direction, it is not hard to see
  that if there is some input election for which by a $C_1+\cdots+C_k$
  action we can achieve our desired change (of creating or removing
  unique-winnerhood for $p$, depending on the case), then there is
  \emph{some} election (not necessarily our input election) for which
  one of those actions alone achieves our desired change. In essence,
  we can view a control action $A$ of type $C_1+\cdots+C_k$ as a
  sequence of operations, each one of one of the $C_1,\ldots, C_k$
  types, that---when executed in order---transform our input election
  into an election where our goal is satisfied. Thus there is a single
  operation within $A$---and this operation is of one of the types
  $C_1, \ldots, C_k$---that transforms some election $E'$ where our
  goal is not satisfied to some election $E''$ where the goal is
  satisfied.
\end{proof}

In the next theorem we show that if a given election system is
vulnerable to some basic type of control and is immune to another basic type of control, then it
is vulnerable to these two types of control combined. The proof of
this theorem is easy, but we need to be particularly careful as
vulnerabilities and immunities can behave quite unexpectedly. For
example, it seems that we can assume that if an election system is
vulnerable to AV and DV then it should also be vulnerable to BV,
because bribing a particular voter can be viewed as first deleting
this voter and then adding---in his or her place---a voter with the
preference order as required by the briber. (This assumes we
have such a voter among the voters we can add, but when arguing
susceptibility/immunity we can make this assumption.) However, there
is an easy election system that is vulnerable to both AV and DV control, but
that is immune to BV control. This system simply says that in an election
$E=(C,V)$, where $C = \{c_1, \ldots, c_m\}$ and $V = (v_1, \ldots ,
v_n)$, the winner is the candidate $c_i$ such that $n \equiv i-1
\pmod{m}$.\footnote{Of course, this election system is not neutral;
permuting the names of the candidates can change the outcome of an election.}

\begin{theorem}\label{thm:propagate-vulnerability}
  Let $\electionsystem$ be an election system and let
  $C_1+\cdots+C_k+D_1+\cdots+D_\ell$ be a variant of multiprong
  control (so $1 \leq k,\ell \leq 5$ and each $C_i$ and each $D_i$ is
  a basic control type) such that $\electionsystem$ is vulnerable to
  constructive (destructive) $C_1+\cdots+C_k$ control but is immune to
  constructive (destructive) $D_1+\cdots+D_\ell$
  control. $\electionsystem$ is vulnerable to
  $C_1+\cdots+C_k+D_1+\cdots+D_\ell$ control.
\end{theorem}
\begin{proof}
  We will give a proof for the constructive case only. The proof for
  the destructive case is analogous.  Let $\electionsystem$ be an
  election system as in the statement of the theorem and let $I$ be an
  instance of constructive
  $\electionsystem$-$C_1+\cdots+C_k+D_1+\cdots+D_\ell$ control, which
  contains election $E = (C,V)$, information about the specifics of
  control actions we can implement, and where the goal is to ensure
  that candidate $p$ is a unique winner.  Let us first consider the
  case where BV is not among $C_1, \ldots, C_k, D_1, \ldots D_\ell$.

  Let us assume that there is a sequence $A$ of control actions of
  types $C_1, \ldots C_k, D_1, \ldots, D_\ell$, such that (a) applying
  $A$ to $E$ is a legal control action within $I$, and (b) applying
  $A$ to $E$ results in an election $E_{C+D}$ where $p$ is the unique
  winner.  (We take $A$ to be an empty sequence if $p$ is a unique
  winner of $E$.) We split the sequence $A$ into a subsequence $A_C$
  that contains exactly the actions of types $C_1, \ldots, C_k$, and a
  subsequence $A_D$ that contains exactly the actions of types $D_1,
  \ldots, D_\ell$. Since BV is not among our control actions, it is
  easy to see that it is possible to apply actions $A_C$ to election
  $E$ to obtain some election $E_C$.  (To see why it is important that
  we do not consider BV, assume that BV is among control types $C_1,
  \ldots, C_k$ and AV is among control types $D_1, \ldots, D_\ell$.
  In this case, $A_C$ might include an action
  that bribes a voter that is added by an action from $A_D$.)

  We claim that $p$ is a unique winner of $E_C$. For the sake of
  contradiction, let us assume that this is not the case (note that
  this implies that $p$ is not a unique winner of $E$). If we apply
  control actions $A_D$ to $E_C$, we reach exactly election $E_{C+D}$,
  where $p$ is the unique winner. Yet, this is a contradiction,
  because we assumed that $\electionsystem$ is immune to
  $D_1+\cdots+D_\ell$, i.e., that there is no scenario where control
  action of type $D_1+\cdots+D_\ell$ makes some candidate a unique
  winner if he or she was not a unique winner before.

  Thus, it is possible to ensure that $p$ is a unique winner by
  actions of type $C_1+\cdots+C_k$ alone. We chose $I$ arbitrarily,
  and thus any instance of
  $\electionsystem$-$C_1+\cdots+C_k+D_1+\cdots+D_\ell$ control can be
  solved by an algorithm that considers control actions of type
  $C_1+\cdots+C_k$ only. This proves that $\electionsystem$ is
  vulnerable to $C_1+\cdots+C_k+D_1+\cdots+D_\ell$ control because, as
  we have assumed, it is vulnerable to $C_1+\cdots+C_k$ control.

  It remains to prove the theorem for the case where BV \emph{is}
  among our control actions. In the case where BV is among the control
  actions but AV is not, or if AV and BV are in the same group of
  actions (i.e., either both are among the $C_i$'s or both are among the
  $D_i$'s), it is easy to see that the above proof still works.
  Similarly, if BV is among the $D_i$'s and AV is among the $C_i$'s, the above
  proof works as well.  The only remaining case is if our allowed
  control types include both BV and AV, where BV is among the $C_i$'s and
  AV is among the $D_i$'s.

  In this last case, the proof also follows the general structure of
  the previous construction, except that we have to take care of one
  issue: It is possible that sequence $A_C$ includes bribery of voters
  that are to be added by actions from $A_D$.  (We use the same
  notation as in the main construction.)  Let $V_{\bv}$ be the collection of
  voters that $A_C$ requires to bribe, but that are added in $A_D$.
  We form a sequence $A'_C$ that is identical to $A_C$, except that it
  starts by adding the voters from $V_{\bv}$, and we let $A'_D$ be
  identical to $A_D$, except that it no longer includes adding the
  voters from $V_{\bv}$. Using sequences $A'_C$ and $A'_D$ instead of
  $A_C$ and $A_D$, it is easy to show the following: If it is possible
  to ensure that $p$ is a unique winner in instance $I$ by a legal
  action of type $C_1+\cdots+C_k+D_1+\cdots+D_\ell$, then it is also
  possible to do so by a legal action of type $C_1+\cdots+C_k+AV$,
  where each added voter is also bribed.  Thus, given an instance $I$
  of $\electionsystem$-$C_1+\cdots+C_k+D_1+\cdots+D_\ell$ we can solve
  it using the following algorithm. Let $W$ be the collection of
  voters that can be added within $I$ and let $k_{\av}$ be the limit
  on the number of voters that we can add.
  \begin{enumerate}
  \item Let $t$ be $\max(k_\av, \|W\|)$.
  \item For each $i$ in $\{0, 1, \ldots, t\}$ execute the next two substeps.
    \begin{enumerate}
    \item Form instance $I'$ that is identical to $I$, except $i$
      (arbitrarily chosen) voters from $W$ are added to the election.
    \item Run the $\electionsystem$-$C_1+\cdots+C_k$ algorithm on
      instance $I'$ and accept if it does.
    \end{enumerate}
  \item If the algorithm has not accepted yet, reject.
  \end{enumerate}
  It is easy to see that this algorithm is correct and, since
  $\electionsystem$ is vulnerable to $C_1+\cdots+C_k$, works in
  polynomial time.  This completes the proof of the theorem.
\end{proof}

\begin{theorem}\label{thm:propagate-resistance}
  Let $\electionsystem$ be an election system and let
  $C_1+\cdots+C_k$, $1 \leq k \leq 5$, be a variant of multiprong
  control. If for some $i$, $1 \leq i \leq k$, $\electionsystem$ is
  resistant to constructive (destructive) $C_i$ control, then
  $\electionsystem$ is resistant to constructive (destructive)
  $C_1+\cdots+C_k$ control.
\end{theorem}
\begin{proof}
Let $C_i$ be the control type to which $\electionsystem$ is
resistant.  Since $\electionsystem$ is susceptible
to constructive (destructive) $C_i$ control,
it follows by Theorem~\ref{thm:propagate}
that $\electionsystem$ is
  susceptible to constructive (destructive) $C_1+\cdots+C_k$ control.
And since the $\electionsystem$-$C_i$ constructive (destructive)
  control problem is essentially (give or take syntax) an embedded
  subproblem of the $\electionsystem$-$C_1+\cdots+C_k$ control problem,
  clearly $\electionsystem$ is resistant to $C_1+\cdots+C_k$ control.
\end{proof}

By combining the above three theorems, we obtain a simple tool
that allows us to classify a large number of multiprong control
problems based on the properties of their prongs.

\begin{cor}\label{combining-prongs}
  Let $\electionsystem$ be an election system and let $C_1+\cdots+
  C_k$, $1 \leq k \leq 5$, be a variant of multiprong control, such
  that for each $C_i$, $1 \leq i \leq k$, $\electionsystem$ is 
  resistant, vulnerable, or immune to constructive (destructive) $C_i$
  control.  If there is an $i$, $1 \leq i \leq k$, such that
  $\electionsystem$ is resistant to constructive (destructive) $C_i$ control
  then $\electionsystem$ is resistant to constructive (destructive)
  $C_1+\cdots+C_k$ control. Otherwise, if there is an $i$, $1 \leq i
  \leq k$, such that $\electionsystem$ is vulnerable to constructive
  (destructive) $C_i$ control then $\electionsystem$ is vulnerable to
  constructive (destructive) $C_1+\cdots+C_k$ control. Otherwise,
  $\electionsystem$ is immune to constructive (destructive)
  $C_1+\cdots+C_k$ control.
\end{cor}

Theorem~\ref{thm:propagate-resistance} immediately yields many
``free'' resistance results based on the previous work on control.
However, we will focus on the more interesting issue of proving that even
multiprong control is easy for some election systems whose control has
already been studied (in Section~\ref{sec:combining})
and for candidate control in maximin (Section~\ref{sec:maximin}).

In general, we do not consider partition cases of control in this
paper.  However, we make an exception for the next example, which shows
how even types of control to which a given election system is immune
may prove useful in multiprong control. In constructive control by
partition of candidates (reminder: this is not a basic control
type) in the ties-eliminate model (PC-TE control type), we are given
an election $E = (C,V)$ and a preferred candidate $p \in C$, and we
ask whether it is possible to 
find a partition $(C_1,C_2)$ of $C$ (i.e.,
$C_1 \cup C_2 = C$ and $C_1 \cap C_2 = \emptyset$)
such that $p$ is a winner of the following two-round
election: We first find the winner sets, $W_1$ and $W_2$, of elections
$(C_1,V)$ and $(C_2,V)$. If $W_1$ ($W_2$) contains more than
one candidate, we set $W_1 = \emptyset$ ($W_2 = \emptyset$), since
we are in the ``ties eliminate'' model.  The
candidates who win election $(W_1 \cup W_2,V)$ are the winners of the
overall two-stage election.

Now, let us look at constructive approval-AC+PC-TE control, where (by
definition, let us say) we first add new candidates and then perform
the partition action. We consider an approval election with two
candidates, $p$ and $c$, where $p$ has 50 approvals and $c$ has
100. We are also allowed to add candidate $c'$, who has 100
approvals. Clearly, it is impossible to make $p$ a unique winner by adding
$c'$. Exercising the partition action alone does not ensure $p$'s
victory either. However, combining both AC and PC-TE does the job.  If
we first add $c'$ to the election and then partition candidates into
$\{p\}$ and $\{c,c'\}$ then, due to the ties-eliminate rule, $p$
becomes the unique winner. It is rather interesting that even though approval
is immune to constructive AC control, there are cases where one has to
apply AC control to open the possibility of effectively using other
types of control. 

The above example is perhaps surprising in light of 
Theorem~\ref{thm:propagate-vulnerability}. In essence, in the proof of
that theorem we argue that if an election system is vulnerable to some
basic control type $C$ but is immune to some other basic control type $D$, then it is
also vulnerable to control type $C+D$. We proved the theorem by
showing that we can safely disregard the actions of type $D$ (assuming
$C$ does not include BV control type). The above example shows that
this proof approach would not work if we considered PC-TE in addition
to the basic control types.

\subsection{Combining Vulnerabilities}
\label{sec:combining}
In the previous section we considered the case where separate
prongs of a multiprong control problem have different computational
properties, e.g., some are resistant, some are vulnerable, and some
are immune. In this section we consider the case where an election
system is vulnerable to each prong separately, and we show how such
vulnerabilities combine within election systems for which control
results were obtained in previous papers (see Table~\ref{tab:results}).
In particular, in the next
theorem we show that for all the election systems considered
in~\cite{bar-tov-tri:j:control},
\cite{hem-hem-rot:j:destructive-control},
and~\cite{fal-hem-hem-rot:j:llull}, all constructive vulnerabilities
to AC, DC, AV, DV, and BV combine to vulnerabilities, and all
destructive vulnerabilities to AC, DC, AV, DV, BV combine to
vulnerabilities.\footnote{Constructive bribery for plurality and
  constructive bribery for approval have been considered
  in~\cite{fal-hem-hem:j:bribery} and constructive and destructive
  bribery for Copeland has been studied
  in~\cite{fal-hem-hem-rot:j:llull}. In Theorem~\ref{thm:combine}
  we---in effect---give polynomial-time algorithms for destructive
  bribery in plurality, approval, and Condorcet. Constructive
  Condorcet-BV is $\np$-complete and this is implicitly shown
  in~\cite[Theorem~3.2]{fal-hem-hem-rot:j:llull}.}  That is, for each
election system studied in these three papers, if it is separately
vulnerable to some basic control types $C_1, \ldots, C_k$, where each
$C_i \in \{$AC,\,DC,\,AV,\,DV,\,BV$\}$, it is also vulnerable to
$C_1+\cdots+C_k$.

\begin{theorem}\label{thm:combine}
  (a) Plurality is vulnerable to both constructive AV+DV+BV control and
  destructive AV+DV+BV control.  (b) Both Condorcet and approval are
  vulnerable to AC+AV+DV+BV destructive control.  (c) For each rational
  $\alpha$, $0 \leq \alpha \leq 1$, Copeland$^\alpha$ is vulnerable to
  destructive AC+DC control.
\end{theorem}
\begin{proof}
  (a) Let us consider an instance $I$ of constructive plurality-AV+DV+BV
  control where we want to ensure candidate $p$'s victory: It is
  enough to add all the voters who vote for $p$ (or as many as we are
  allowed) and then, in a loop, keep deleting voters who vote
  for a candidate other than $p$ with the highest score,
  until $p$ is the only candidate with the highest score
  or we have exceeded our limit of
  voters to delete. Finally, in a loop, keep bribing voters who
  vote for a candidate other than $p$ with the highest score
  to vote for $p$, 
  until $p$ is the only candidate with the highest score
  or we have exceeded our limit of voters to bribe.
  If $p$ becomes a unique winner via this procedure, then
  accept. Otherwise reject. We omit the easy proof for the destructive
  case.

  (b) Let $I$ be an instance of destructive Condorcet-AC+AV+DV+BV, where
  our goal is to prevent candidate $p$ from being a Condorcet winner
  (we assume that $p$ is a Condorcet winner before any control action
  is performed). It is enough to ensure that some candidate $c$ wins a
  head-to-head contest with $p$. Our algorithm works as follows.

  Let $C$ be the set of candidates originally in the election and let
  $A$ be the set of candidates that we can add (we take $A =
  \emptyset$ if we are not allowed to add any candidates). For each $c
  \in (C \cup A) - \{p\}$ we do the following:
  \begin{enumerate}
  \item Add as many voters who prefer $c$ to $p$ as possible.
  \item Delete as many voters who prefer $p$ to $c$ as possible.
  \item Among the remaining voters who prefer $p$ to $c$, bribe
        as many as possible to rank $c$ first.
  \end{enumerate}
  If after these actions $c$ wins his or her head-to-head contest with
  $p$ then we accept.  If no $c \in (C \cup A) - \{p\}$ leads to
  acceptance, then we reject. It is easy to see that this algorithm is
  correct and runs in polynomial time. (We point out that it is enough
  to add only a single candidate, the candidate $c$ that prevents $p$
  from winning, if he or she happens to be a member of $A$).

  For the case of approval, our algorithm works similarly, except the
  following differences: We add voters who approve of $c$ but not of
  $p$. We delete voters who approve of $p$ but not of $c$.  For each
  remaining voter $v_i$, if we still have not exceeded our bribing
  limit, if $v_i$ approves of $p$ but not of $c$, we bribe $v_i$ to
  reverse approvals on $p$ and $c$. (Note that if we do not exceed our
  bribing limit by this procedure, this means that each voter that
  approves of $p$ also approves of $c$ and thus $p$ is not a unique
  winner.)  If these actions lead to $p$ not being a unique winner, we
  accept.  If we do not accept for any $c \in (C \cup A) - \{p\}$, we reject.

  (c) The idea is to combine Copeland$^\alpha$ destructive-AC and
  destructive-DC algorithms~\cite{fal-hem-hem-rot:j:llull}.  We give
  the full proof for the sake of completeness.

  Let us fix a rational value $\alpha$, $0 \leq \alpha \leq 1$. Given
  an election $E$ and a candidate $c$ in this election, we write
  $\score_E^\alpha(c)$ to denote $\copeland^\alpha$ score of $c$.  Let
  $I$ be an instance of destructive $\copeland^\alpha$-AC+DC control,
  with an election $E = (C,V)$, where we can add at most $k_\ac$
  spoiler candidates from the set $A$, and where we can delete at most
  $k_\dc$ candidates. Our goal is to ensure that some despised
  candidate $d \in C$ is not a unique winner. Our algorithm is based
  on the following simple observation of Faliszewski et
  al.~\cite{fal-hem-hem-rot:j:llull}. For each candidate $c \in C$:
  \[ 
     \score_{(C,V)}^\alpha(c) = \sum_{c' \in C - \{c\}} \score^\alpha_{(\{c,c'\},V)}(c).
  \]
  Our goal is to prevent candidate $d$ from being a unique winner. If
  $d$ is not a unique winner, we immediately accept. Otherwise, we
  seek a candidate $c \in C \cup A$ such that we can ensure that $c$'s
  score is at least as high as that of $d$. Thus, for each $c \in C \cup A$
  we do the following.
  \begin{enumerate}
    \item If $c \in A$, and $k_\ac > 0$, we add $c$ to the election (and
          if $c \in A$ but $k_\ac = 0$, we proceed to the next $c$).

    \item As long as we can still add more candidates, we keep
      executing the following operation: If there is a candidate $c'
      \in A$ such that value $a(c') = \score^\alpha_{(\{c,c'\},V)}(c) -
      \score^\alpha_{(\{d,c'\},V)}(d)$ is positive, we add a candidate
      $c'' \in A$, for whom $a(c'')$ is highest.

    \item As long as we can still delete candidates, we keep executing
      the following operation: If there is a candidate $c' \in C$ such
      that value $r(c') = \score^\alpha_{(\{d,c'\},V)}(d) -
      \score^\alpha_{(\{c,c'\},V)}(c)$ is positive, we delete a
      candidate $c'' \in C$, for whom $r(c'')$ is highest.

    \item If after these steps $d$ is not a unique winner, we accept.
  \end{enumerate}
  If we do not accept for any $c \in C \cup A$, we reject. 

  It is easy to see that we never delete a candidate that we have
  added. Also, it is easy to see that the algorithm works in
  polynomial time, and that it is correct. Correctness follows from
  the fact that (a) in the main loop of the algorithm, when dealing
  with candidate $c \in C \cup A$, each addition of a candidate and
  each deletion of a candidate increases the difference between the
  score of $c$ and the score of $d$ as much as is possible, and (b)
  the order of adding/deleting candidates is irrelevant.
\end{proof}

As witnessed
by Theorem~\ref{thm:combine} and the results of
Section~\ref{sec:maximin}, for all natural election systems that we
have considered, all constructive vulnerabilities combine and so do all
destructive ones. It is natural to wonder whether this is a necessary
consequence of our model of multiprong control or whether in fact there is
an election system for which combining two
control types to which the system is vulnerable yields a multipronged
control problem to which the system is resistant.
Theorem~\ref{thm:origllull} shows that the latter is the case,
even for a natural (though rather unusual) election system.

In the thirteenth century, Ramon Llull proposed an election system
that could be used to choose 
popes and leaders of monastic orders
(see~\cite{hae-puk:j:electoral-writings-ramon-llull,mcl-lor:t:papacy}).
In his system, voters choose the winner from among themselves
(so, the candidates are the same as the voters).
Apart from that, Llull's voting system is basically $\copeland^1$,
the version of Copeland that most richly rewards ties.
Formally, we define the voting system $\origllull$ as follows: For 
an election $E = (C,V)$, if the set of names of $V$, which we will
denote by $\names(V)$, is not
equal to $C$, then there are no winners.  Otherwise, a candidate
$c \in C$ is a winner if and only if it is a $\copeland^1$ winner.
Note that single-prong AC and AV control for 
$\origllull$ don't make all that much sense, and so it should
come as no surprise that 
$\origllull$ is vulnerable to both constructive AC control 
and constructive AV control.  In addition, we will show 
(by renaming and padding) that $\copeland^1$-AV can be 
be reduced to 
$\origllull^1$-AC+AV.
Since $\copeland^1$ is resistant to constructive
control by adding voters~\cite{fal-hem-hem-rot:j:llull},
this then leads to the following  theorem.

\begin{theorem}\label{thm:origllull}
$\origllull$ is vulnerable to both constructive AC control 
and constructive AV control but is resistant to
constructive AC+AV control.
\end{theorem}
\begin{proof}

It is immediate that $\origllull$ is susceptible to constructive
AC, AV, and (by Theorem~\ref{thm:propagate}) AC+AV control.
  It is also easy to see that constructive
  $\origllull$-AC (AV) control is in $\p$: If possible add
  candidates (voters) such that the set of voter names is equal to
  the set of candidates, and then check if the preferred candidate
  is a unique $\copeland^1$ winner.  If this is not possible,
  reject.

  We will now show, via a reduction from constructive
  $\copeland^1$-AV control (which is NP-hard~\cite{fal-hem-hem-rot:j:llull})
 that constructive $\origllull$-AC+AV control is $\np$-hard.
Let $C$ be a set of candidates, $V$ and $W$ be two disjoint collections
of voters with preference lists over $C$, $p \in C$ the preferred candidate,
and $k \in \naturals$.  The question is whether there exists a 
subcollection $W' \subseteq W$ of size at most $k$ such that
$p$ is a unique $\copeland^1$ winner of $(C,V \cup W')$.
Without loss of generality, we assume that $V$ is not empty.

We will now show how to pad this election. For an $\origllull$ election
to be non-trivial, we certainly need to have the same number of candidates
as voters (later, we will also rename the voters so that they are the same
as the candidates).
If $\|V\| < \|C\|$, we want
to add a collection of new dummy voters $V'$ such that
$\|V\| + \|V'\| = \|C\|$ and such that adding $V'$ to an election
does not change the relative $\copeland^1$ scores of the candidates.
This can be accomplished by letting half of the voters in $V'$
vote $C$ (recall Convention~A)
and half of the voters in $V'$ vote $\revnot{C}$.
Of course, this can only be done if $\|V'\|$ is even.  

So, we will do the following.
If $\|V\| < \|C\|$, we add a collection of new voters $V'$ such that
$\|V'\| = \|C\| - \|V\|$ if $\|C\| - \|V\|$ is even, and
$\|V'\| = \|C\| - \|V\| + 1$ if $\|C\| - \|V\|$ is odd.
If $\|V\| \geq \|C\|$, we let $V' = \emptyset$.
Half of the voters in $V'$
vote $C$ and half of the voters in $V'$ vote $\revnot{C}$.
In addition, we introduce a set $A$ of new candidates such that
$\|C\| + \|A\| = \|V\| + \|V'\| + \|W\|$.  Note that this is always
possible, since $\|V\| + \|V'\| \geq \|C\|$. 
We extend the votes of the voters
(in $V$, $V'$, and $W$) to $C \cup A$
by taking their preference order on $C$ and following
this by the candidates in $A$ in some fixed, arbitrary order.  Note that
this will have the effect that candidates in $A$ will never be winners.

Let $W' \subseteq W$, $A' \subseteq A$,
$E = (C,V \cup W')$, 
$E' = (C \cup A', V \cup V' \cup W')$.
It is easy to see that the following hold (recall that $V$ is not empty).
\begin{enumerate}
\item For all  $d \in A'$, 
$\score_{{E'}}^1(d) \leq \|A'\| - 1$.

\item  For all  $c \in C$,
$\score_{{E'}}^1(c) =
\score_{E}^1(c) + \|{A'}\|$.

\item  For all  $c,c' \in C, c \neq c'$,
$\score_{E}^1(c) - \score_{E}^1(c') = \score_{{E'}}^1(c) -
\score_{{E'}}^1(c')$. 

\item
\label{item:origllull}
$p$ is a unique $\copeland^1$ winner of $E$ if and only if
$p$ is a unique $\copeland^1$ winner of $E'$.
\end{enumerate}

We are now ready to define the reduction.
Name the voters such that
$\names(V \cup V') \supseteq C$ and
$\names(V \cup V' \cup W) = C \cup A$.
Then map $(C, V, W, p, k)$ to 
$(C, A, V \cup V', W, p, \|A\|, k)$.
We claim that $p$ can be made a unique $\copeland^1$ winner
of $(C,V)$ by adding
at most $k$ voters from $W$ if and only if
$p$ can be made a unique $\origllull$ winner of 
$(C, V \cup V')$ by adding (an unlimited number of) candidates from $A$ and at most
$k$ voters from  $W$.

First suppose that $W'$ is a subcollection of $W$ of size at most $k$
such that $p$ is the unique $\copeland^1$ winner of
$(C, V \cup W')$.  Let $A' \subseteq A$ be
the set of candidates such that $C \cup A' 
= \names(V \cup V' \cup W')$. By item~\ref{item:origllull} above,
$p$ is the unique $\copeland^1$ winner of 
$(C \cup A' , V \cup V' \cup W')$, and thus
$p$ is the unique $\origllull$ winner of 
$(C \cup A', V \cup V' \cup W')$.

For the converse, suppose that there exist
$A' \subseteq A$ and
$W' \subseteq W$ such that $\|W'\| \leq k$, 
and $p$ is the unique $\origllull$ winner of 
$(C \cup A' , V \cup V' \cup W')$.
Then $p$ is the unique $\copeland^1$ winner of 
$(C \cup A' , V \cup V' \cup W')$, and,
by item~\ref{item:origllull}, $p$ is the unique $\copeland^1$ winner of
$(C, V \cup W')$.

Thus our reduction is correct and, since
it can be computed in polynomial time, the proof is complete.
\end{proof}

$\origllull$ is neutral (permuting the names of the candidates
does not affect the outcome of an election) but not anonymous
(renaming the voters can change the outcome of an election).
By sneakily building the preference orders of the voters into the
names of the candidates, we can make the
system anonymous as well as neutral (at the price of losing
naturalness).

\begin{theorem}\label{thm:origllullanonymous}
There exists a neutral and anonymous election system
$\electionsystem$ such that $\electionsystem$ is vulnerable to both
constructive AC control 
and constructive AV control but is resistant to constructive AC+AV control.
\end{theorem}
\begin{proof}
We first describe $\electionsystem$.  On input $(C,V)$, an election, if
there exists a set $I \subseteq \naturals^+$ and bijections
$c$ from $I$ to $C$ and
$v$ from $I$ to $V$ such that for all $i \in I$,
$c(i) = (i,>_i)$ where $>_i$ is a preference order on $I$
(i.e., we interpret candidate names as pairs consisting of
a positive integer and a preference order on $I$) and
voter $v(i)$ corresponds to candidate $c(i)$ in the sense that
for all $j,k \in I$,
$j >_i k$ if and only if $c(j) > c(k)$ in voter $v(i)$'s preference order,
then the winners are exactly the $\copeland^1$ winners.  Otherwise,
there are no winners.

Note that $\electionsystem$ is neutral and anonymous and basically
the same as $\origllull$.
The same argument as used for $\origllull$ in the proof
of Theorem~\ref{thm:origllull} shows that $\electionsystem$ is vulnerable
to constructive AC and AV control and susceptible to AC+AV control.
To show that constructive $\electionsystem$-AC+AV control is $\np$-hard,
we adapt the reduction from from constructive
$\copeland^1$-AV control to
constructive $\origllull$-AC+AV control from the proof
of Theorem~\ref{thm:origllull}.
Let $C$ be a set of candidates, $V$ and $W$ be two disjoint collections
of voters with preference lists over $C$, $p \in C$ the preferred candidate,
and $k \in \naturals$. 
Without loss of generality, we assume that $V$ is not empty.
Let $V'$ and $A$ be as in the proof of Theorem~\ref{thm:origllull}.
Recall that $\|V \cup V'\| \geq \|C\|$ and
$\|V \cup V' \cup W\|  = \|C \cup A\|$.
From the proof of Theorem~\ref{thm:origllull} we have the following.
\begin{claim}
\label{cl:winner}
Let $W' \subseteq W$, $A' \subseteq A$,
$E = (C,V \cup W')$, 
$E' = (C \cup A', V \cup V' \cup W')$.
$p$ is a unique $\copeland^1$ winner of $E$ if and only if
$p$ is a unique $\copeland^1$ winner of $E'$.
\end{claim}

We are now ready to define the reduction.
We will first rename the candidates.  Note that renaming
candidates does not change the outcome of a $\copeland^1$ election.
Number the candidates in $C \cup A$ from $1$ to $\|C \cup A\|$ such that the
candidates in $C$ are numbered from $1$ to $\|C\|$.
Number the voters in $V \cup V' \cup W$ from $1$ to $\|V \cup V' \cup W\|$
(= $\|C \cup A\|$)
such that the voters in $\|V \cup V'\|$ are numbered
from $1$ to $\|V \cup V'\|$.
Now rename candidate $i$ to
$c_i = (i,>_i)$ where $>_i$ is the preference order on
$\{1, \ldots, \|C\| + \|A\|\}$ such that
for all $j, k \in \{1, \ldots, \|C\| + \|A\|\}$,
$j >_i k$ if and only if $j > k$ in voter $i$.
Rename all candidates occurring in $C$, $A$, $V$, $V'$, and $W$ in this way.
We claim that $p$ can be made a unique $\copeland^1$ winner
of $(C,V)$ by adding
at most $k$ voters from $W$ if and only if
$p$ can be made a unique $\electionsystem$ winner of 
$(C, V \cup V')$ by adding candidates from $A$ and at most
$k$ voters from  $W$.

First suppose that $W'$ is a subcollection of $W$ of size at most $k$
such that $p$ is the unique $\copeland^1$ winner of
$(C, V \cup W')$.  Let $A' \subseteq A$ be
the set of candidates such that 
$C \cup A' = \{c_i \ | \ \mbox{voter $i$ is in $V \cup V' \cup W'$}\}$.
By Claim~\ref{cl:winner},
$p$ is the unique $\copeland^1$ winner of 
$(C \cup A' , V \cup V' \cup W')$, and thus
$p$ is the unique $\electionsystem$ winner of 
$(C \cup A', V \cup V' \cup W')$.

For the converse, suppose that there exist
$A' \subseteq A$ and
$W' \subseteq W$ such that $\|W'\| \leq k$, 
and $p$ is the unique $\electionsystem$ winner of 
$(C \cup A' , V \cup V' \cup W')$.
Then
$p$ is the unique $\copeland^1$ winner of 
$(C \cup A' , V \cup V' \cup W')$, and
by Claim~\ref{cl:winner},
$p$ is the unique $\copeland^1$ winner of $(C, V \cup W')$.

Thus our reduction is correct and, since
it can be computed in polynomial time, the proof is complete.
\end{proof}

\section{Control in Maximin}
\label{sec:maximin}

In this section we initiate the study of control in the maximin
election system. Maximin is 
loosely related to $\copeland^\alpha$  voting in the sense that
both are defined in terms of the pairwise head-to-head
contests.  In addition, the unweighted coalitional
manipulation problem for maximin and $\copeland^\alpha$ ($\alpha \neq 0.5$)
exhibits the same unusual behavior:
It is in P for one manipulator and
NP-complete for two or more
manipulators~\cite{con-pro-ros-xia:c:unweighted-manipulation,fal-hem-sch:c:copeland-ties-matter,fal-hem-sch:c:copeland01}.
Thus one might wonder whether both systems will be
similar with regard to 
their resistances to control.  In fact, there are very interesting
differences.

It is easy to see that maximin is susceptible to all basic types of
constructive and destructive control.
And so, by Theorem~\ref{thm:propagate}, to show vulnerability to constructive
(destructive) $\calC$ control it suffices to 
give a polynomial-time algorithm that decides the constructive
(destructive) $\electionsystem$-$\calC$ problem,
and to show
resistance to constructive
(destructive) $\calC$ control it suffices to show
that the constructive (destructive) $\electionsystem$-$\calC$ problem is
$\np$-hard.

\subsection{Candidate Control in Maximin}
Let us now focus on candidate control in maximin, that is, on AC,
AC$_\mathrm{u}$, and DC control types, both in the constructive and in
the destructive setting.
As is the case for Copeland$^\alpha$, $0
\leq \alpha \leq 1$, maximin is resistant to control by adding
candidates.
\begin{theorem}
  Maximin is resistant to constructive AC control.
\end{theorem}
\begin{proof}
We give a reduction from $\xthreec$. Let $(B,\calS)$,
  where $B = \{b_1, \ldots, b_{3k}\}$ is a set of $3k$ elements and
  $\calS = \{S_1, \ldots, S_n\}$ is a set of $3$-subsets of $B$, be our
  input $\xthreec$ instance. We form an election $E = (C \cup A,V)$,
  where $C = B \cup \{p\}$, $A = \{a_1, \ldots, a_n\}$, and $V = (v_1,
  \ldots, v_{2n+2})$. (Candidates in $A$ are the spoiler candidates,
  which the attacker has the ability to add to election $(C,V)$.)

  Voters in $V$ have the following preferences. For each $S_i \in
  \calS$, voter $v_i$ reports preference list $p > B - S_i > a_i > S_i
  > A - \{a_i\}$ and voter $v_{n+i}$ reports preference list
  $\revnot{A - \{a_i\}} > a_i > \revnot{S_i} > \revnot{B-S_i} >
  p$. Voter $v_{2n+1}$ reports $p > A > B$ and voter $v_{2n+2}$
  reports $\revnot{B} > p > \revnot{A}$.

  We claim that there is a set $A' \subseteq A$ such that $\|A'\| \leq
  k$ and $p$ is a unique winner of $(C \cup A',V)$ if and only if
  $(B,\calS)$ is a ``yes''-instance of $\xthreec$.

  To show the claim, let $E' = (C,V)$. For each pair of distinct
  elements $b_i,b_j \in B$, we have that $N_{E'}(b_i,b_j) = n + 1$,
  $N_{E'}(p,b_i) = n+1$, and $N_{E'}(b_i,p) = n+1$. That is, all
  candidates in $E'$ tie. Now consider some set $A'' \subseteq A$,
  $\|A''\| \leq k$, and an election $E'' = (C \cup A'',V)$. Values of
  $N_{E''}$ and $N_{E'}$ are the same for each pair of candidates in
  $\{p\} \cup B$. For each pair of distinct elements $a_i, a_j \in
  A''$, we have $N_{E''}(p,a_i) = n+2$, $N_{E''}(a_i,p) = n$, and
  $N_{E''}(a_i,a_j) = n+1$.  For each $b_i \in B$ and each $a_j \in
  A''$ we have that
  \[ N_{E''}(b_i,a_j) = \left\{ \begin{array}{ll}
      n   & \mbox{ if $b_i \in S_j$,} \\
      n+1 & \mbox{ if $b_i \notin S_j$,}
  \end{array}\right. 
  \]
  and, of course, $N_{E''}(a_j,b_i) = 2n+2 - N_{E''}(b_i,a_j)$. Thus,
  by definition of maximin, we have the following scores in $E''$: (a)
  $\score_{E''}(p) = n+1$, (b) for each $a_j \in A''$,
  $\score_{E''}(a_j) = n$, and (c) for each $b_i \in B$,
  \[
  \score_{E''}(b_i) = \left\{ \begin{array}{ll}
      n    & \mbox{ if $(\exists a_j \in A'')[b_i \in S_j]$,} \\
      n+1  & \mbox{ otherwise. } \\
    \end{array}\right.
  \]

  $A''$ corresponds to a family $S''$ of 3-sets from $\calS$ such that
  for each $j$, $1 \leq j \leq n$, $S''$ contains set $S_j$ if and
  only if $A''$ contains $a_j$.  Since $\|A''\| \leq k$, it is easy to
  see that $p$ is a unique winner of $E''$ if and only if $S''$ is
  an exact cover of $B$.
\end{proof}

Copeland$^\alpha$, $0 \leq \alpha \leq 1$, is resistant to
constructive AC control, but for $\alpha \in \{0,1\}$,
Copeland$^\alpha$ is vulnerable to constructive control by adding an
unlimited number of candidates. It turns out that so is
maximin. However, interestingly, in contrast to
Copeland,  maximin is also vulnerable to DC
control, and in fact even to AC$_\mathrm{u}$+DC control.  Intuitively,
in constructive AC$_\mathrm{u}$+DC control we should add as many
candidates as possible (because adding a candidate generally decreases
other candidates' scores, making our preferred candidate's way to
victory easier) and then delete those candidates who stand in our
candidate's way (i.e., those whose existence blocks the preferred candidate's score
from increasing). Studying constructive AC$_\mathrm{u}$+DC control for
maximin jointly leads to a compact, coherent algorithm. If we were to
consider both control types separately, we would have to give two
fairly similar algorithms while obtaining a weaker result.

\begin{theorem}\label{thm:maximin:acu+dc}
  Maximin is vulnerable to constructive AC$_\mathrm{u}$+DC control.
\end{theorem}
\begin{proof}
  We give a polynomial-time algorithm for constructive
  maximin-AC$_\mathrm{u}$+DC control. The input contains an election
  $E = (C,V)$, a set of spoiler candidates $A$, a preferred candidate
  $p \in C$, and a nonnegative integer $k_\dc$. Voters in $V$ have
  preference lists over the candidates in $C \cup A$.  We ask whether
  there exist sets $A' \subseteq A$ and $C' \subseteq C$ such that (a)
  $\|C'\| \leq k_\dc$ and (b) $p$ is a unique winner of election
  $((C-C')\cup A',V)$. If $k_\dc \geq \|C\|-1$, we accept immediately
  because we can delete all candidates but $p$. Otherwise, we use the
  following algorithm.

  \begin{description}

  \item[Preparation.]  We rename the candidates in $C$ and $A$ so that
    $C = \{p, c_1, \ldots, c_m\}$ and $A = \{c_{m+1}, \ldots,
    c_{m+m'}\}$.  Let $E' = (C \cup A,V)$ and let $P = \{ N_{E'}(p,
    c_i) \mid c_i \in C \cup A\}$. That is, $P$ contains all the
    values that candidate $p$ may obtain as scores upon deleting some
    candidates from $E'$. For each $k \in P$, let $Q(k) = \{c_i \mid
    c_i \in C \cup A - \{p\} \land N_{E'}(p,c_i) < k\}$. Intuitively, $Q(k)$
    is the set of candidates in $E'$ that prevent $p$ from having at
    least $k$ points.

  \item[Main loop.] For each $k \in P$, our algorithm tests
    whether by deleting at most $k_\dc$ candidates from $C$ and any
    number of candidates from $A$ it is possible to ensure that $p$
    obtains exactly $k$ points and becomes a unique winner of
    $E'$. Let us fix some value $k \in P$. We build a set $D$ of
    candidates to delete. Initially, we set $D = Q(k)$. It is easy to
    see that deleting candidates in $Q(k)$ is a necessary and
    sufficient condition for $p$ to have score $k$. However, deleting
    candidates in $Q(k)$ is not necessarily sufficient to ensure that
    $p$ is a unique winner because candidates with scores greater or equal to
    $k$ may exist. We execute the following loop (which we will call
    the \emph{fixing loop}): 
    \begin{enumerate}
    \item\label{step:acdc:1} Set $E'' = ((C\cup A)-D,V)$.
    \item\label{step:acdc:2} Pick a candidate $d \in (C \cup A)-D$
      such that $\score_{E''}(d) \geq k$ (break from the loop if no
      such candidate exists).
    \item\label{step:acdc:3} Add $d$ to $D$ and jump back to
      Step~\ref{step:acdc:1}.
    \end{enumerate}
    \noindent
    We accept if $C \cap D \leq k_\dc$ and we proceed to the next
    value of $k$ otherwise.\footnote{If we accept, $D$ implicitly
      describes the control action that ensures $p$'s victory: We
      should delete from $C$ the candidates in $C \cap D$ and add from
      $A$ the candidates in $A - D$.}  If none of the values $k \in P$
    leads to acceptance then we reject.
  \end{description}

  Let us now briefly explain why the above algorithm is correct. It is
  easy to see that in maximin adding some candidate $c$ to an election
  does not increase other candidates' scores, and deleting some
  candidate $d$ from an election does not decrease other candidates'
  scores. Thus, if after deleting candidates in $Q(k)$ there still are
  candidates other than $p$ with $k$ points or more, the only way to
  ensure $p$'s victory---without explicitly trying to increase $p$'s
  score---is by deleting those candidates. Also, clearly, the only way
  to ensure that $p$ has exactly $k$ points is by deleting candidates
  $Q(k)$.

  Note that during the execution of the fixing loop, the score of $p$
  might increase to some value $k' > k$. If that happens, it means
  that it is impossible to ensure $p$'s victory while keeping his or
  her score equal to $k$. However, we do not need to change $k$ to
  $k'$ in that iteration of the main loop as we will consider $k'$ in
  a different iteration.
\end{proof}

Maximin is also vulnerable to destructive AC+DC control.  The proof
relies on the fact that (a) if there is a way to prevent a despised
candidate from winning a maximin election via adding some spoiler
candidates then there is a way to do so by adding at most $2$ 
candidates, (b) adding a candidate cannot increase the score of any
candidate other than the added one, and (c) deleting a candidate cannot
decrease the score of any candidate other than the deleted one. 
In essence, the algorithm performs a brute-force search for the
candidates to add and then uses the constructive maximin-DC control
algorithm from Theorem~\ref{thm:maximin:acu+dc}.

\begin{theorem}
  Maximin is vulnerable to destructive AC+DC control.
\end{theorem}
\begin{proof}
  We will first give an algorithm for destructive maximin-AC and then
  argue how it can be combined with the algorithm from
  Theorem~\ref{thm:maximin:acu+dc} to solve destructive maximin-AC+DC
  in polynomial time.

  Let us first focus on the destructive AC problem.  Our input is an
  election $E = (C,V)$, where $C = \{d, c_1, \ldots, c_m\}$ and $V =
  (v_1, \ldots, v_n)$, a spoiler candidate set $A = \{c_{m+1}, \ldots,
  c_{m'}\}$, and a nonnegative integer $k_\ac$. The voters have
  preference orders over $C \cup A$. The goal is to ensure that $d$ is
  not a unique winner of $E$ via adding at most $k_\ac$ candidates
  from $A$.

  Let us assume that there exists a set $A' \subseteq A$ such that $d$
  is not a unique winner of election $E' = (C \cup A',V)$. Since $d$
  is not a unique winner of $E'$, there exists some candidate $c' \in
  C \cup A'$ such that $\score_{E'}(c') \geq \score_{E'}(d)$. Also, by
  definition of maximin, there is some candidate $d' \in C \cup A'$
  such that $\score_{E'}(d) = N_{E'}(d,d')$. As a consequence, $d$ is
  not a unique winner of election $E'' = (C \cup \{c',d'\},V)$. The
  reason is that $\score_{E''}(d) = \score_{E'}(d)$ (because both $E'$
  and $E''$ contain $d'$) and $\score_{E''}(c') \geq \score_{E'}(c')$
  (because adding the remaining $A' - \{c',d'\}$ candidates to $E''$
  does not increase $c'$'s score). Thus, to test whether it is
  possible to ensure that $d$ is not a unique winner of $E$, it
  suffices to test whether there is a set $A'' \subseteq A$ such that
  $\|A''\| \leq \min(2,k_\ac)$ and $d$ is not a unique winner of
  $(C\cup A'',V)$. Clearly, this test can be carried out in polynomial
  time.

  Let us now consider the AC+DC case. The input and the goal are the
  same as before, except that now we are also given a nonnegative
  integer $k_\dc$ and we are allowed to delete up to $k_\dc$
  candidates.  We now describe our algorithm.  For each set $\{c',d'\}$
  of up to $2$ candidates, $\{c',d'\} \subseteq (C \cup A)-\{d\}$ we
  execute the following steps.  
  \begin{enumerate}
  \item We check if $\|A \cap \{c',d'\}\| \leq k_\ac$ (and we proceed
    to the next $\{c', d'\}$ if this is not the case).
  \item\label{alg:ac+dc:step2} We compute a set $D \subseteq C -
    \{d,c',d'\}$, $\|D\| \leq k_\dc$, that maximizes
    $\score_{E'}(c')$, where $E' = ((C \cup \{c',d'\})-D,V)$.
  \item If $d$ is not a unique winner of $E' = ((C \cup
    \{c',d'\})-D,V)$, we accept.
  \end{enumerate}
  We reject if we do not accept for any $\{c',d'\} \subseteq (C \cup
  A)-\{d\}$.

  The intended role of $d'$ is to lower the score of $d$ and keep it
  at a fixed level, while, of course, the intended role of $c'$ is to
  defeat $d$. By reasoning analogous to that for the AC case, we can
  see that there is no need to add more than two candidates.  Thus,
  given $\{c',d'\}$, it remains to compute the appropriate set $D$. In
  essence, we can do so in the same manner as in the constructive
  AC+DC case.

  Let $k$ be some positive integer. We set $D(k) = \{c_i \in C -
  \{c',d',d\} \mid N_E(c',c_i) < k \}$ and we pick $D = D(i)$, where
  $i$ is as large as possible (but no larger than $\|V\|$) and $\|D\|
  \leq k_\dc$. Deleting candidates in $D$ maximizes the score of $c'$,
  given that we cannot delete $d$ and $d'$.  It is easy to see that
  this $D$ can be computed in polynomial time. 
\end{proof}

\subsection{Control by Adding and Deleting Voters in Maximin}
In this section we consider the complexity of
constructive and destructive AV and DV control types.
(We will consider bribery, BV, in the next section;  recall that
in this paper, bribery is a basic control type, though 
it is usually treated separately in the literature.)
In the previous section we have seen that maximin is
vulnerable to all basic types of constructive and destructive candidate
control except for constructive control by adding candidates
(constructive AC control).  The situation regarding voter control is
quite different: As shown in the next three theorems, maximin is
resistant to all basic types of constructive and destructive voter
control. 

\begin{theorem}
  Maximin is resistant to constructive and destructive AV control.
\end{theorem}
\begin{proof}
We will first give an $\np$-hardness proof
  for the constructive case and then we will describe how to modify it
  for the destructive case.

  We now give a reduction of the $\xthreec$ problem to the
  constructive maximin-AV problem.  Our input $\xthreec$ instance is
  $(B,\calS)$, where $B = \{b_1, \ldots, b_{3k}\}$ is a set of $3k$
  distinct elements and $\calS = \{S_1, \ldots, S_n\}$ is a family of
  $n$ $3$-element subsets of $B$. Without loss of generality, we assume
  $k \geq 1$. Our reduction outputs the following instance. We have an
  election $E = (C,V)$, where $C = B \cup \{p,d\}$ and $V = (v_1,
  \ldots, v_{4k})$.  There are $2k$ voters with preference order $d >
  B > p$, $k$ voters with preference order $p > B > d$, and $k$ voters
  with preference order $p > d > B$. In addition, we have a collection $W =
  (w_1, \ldots, w_n)$ of unregistered voters, where the $i$'th
  voter, $1 \leq i \leq n$, has preference order
  \[ B - S_i > p > S_i > d.\] We claim that there is a subcollection $W'
  \subseteq W$ such that $\|W'\| \leq k$ and $p$ is a unique winner of
  election $(C,V \cup W')$ if and only if $(B,\calS)$ is a
  "yes"-instance of $\xthreec$.

  It is easy to verify that for each $b_i \in B$ it holds that
  $N_E(p,b_i) = 2k$, and that $N_E(p,d) = 2k$. Thus, $\score_E(p) =
  2k$. Similarly, it is easy to verify that $\score_E(d) = 2k$, and
  that for each $b_i \in B$, $\score_E(b_i) \leq k$.  Let $W''$ be a
subcollection of $W$ such that $\|W''\| \leq k$ and let $E'' = (C,V
  \cup W'')$. For each $b_i \in B$ it holds that $\score_{E''}(b_i)
  \leq 2k$.  Since each voter in $W$ ranks $d$ as the least desirable
  candidate, $\score_{E''}(d) = 2k$. What is $p$'s score in election
  $E''$? If there exists a candidate $b_i \in B$ such that there is no
  voter $w_j$ in $W''$ that prefers $p$ to $b_i$, then
  $\score_{E''}(p) = 2k$ (because $N_{E''}(p,b_i) = 2k$). Otherwise,
  $\score_{E''}(p) \geq 2k+1$.  Thus, $p$ is a unique winner of
  $E''$ if and only if $W''$ corresponds to an exact cover of
  $B$. This proves our claim and, as the reduction is clearly
  computable in polynomial time, concludes the proof for the
  constructive maximin-AC case.

  To show that destructive maximin-AC is $\np$-hard, we use the same
  reduction, except that we remove from $V$ a single voter with
  preference list $p > B > d$, and we set the task to preventing $d$
  from being a unique winner. Removing a $p > B > d$ voter from $V$
  ensures that before we start adding candidates, $d$ has score $2k$
  (and this score cannot be changed), $p$ has score $2k-1$ (and $p$
  needs to get one point extra over each other candidate to increase
  his or her score and prevent $d$ from being a unique winner), and
  each $b_i \in B$ has score $k-1$ (thus, no candidate in $B$ can
  obtain score higher than $2k-1$ via adding no more than $k$
  candidates from $W$). The same reasoning as for the constructive
  case proves that the reduction correctly reduces $\xthreec$ to destructive
  maximin-AV.
\end{proof}

\begin{theorem}
  Maximin is resistant to constructive and destructive DV control.
\end{theorem}
\begin{proof}
  We will first show $\np$-hardness for constructive maximin-DV
  control and then we will argue how to modify the construction to
  obtain the result for the destructive case.

  Our reduction is from $\xthreec$.  Let $(B,\calS)$ be our input
  $\xthreec$ instance, where $B = \{b_1, \ldots, b_{3k}\}$, $\calS =
  \{S_1, \ldots, S_n\}$, and for each $i$, $1 \leq i \leq n$,
  $\|S_i\|=3$.  Without loss of generality, we assume that 
  $n \geq k \geq 3$
(if $n < k$ then $\cal S$ does not contain
  a cover of $B$, and if $k \leq 2$ we can solve the problem by brute force).
We form an election $E = (C,V)$, where $C = B \cup
  \{p,d\}$ and where $V = V' \cup V''$, $V' = (v'_1, \ldots,
  v'_{2n})$, $V'' = (v''_1, \ldots, v''_{2n-k+2})$. For each $i$, $1
  \leq i \leq n$, voter $v'_i$ has preference order
  \[ d > B - S_i > p > S_i\]
  and voter $v'_{n+i}$ has preference order
  \[ d > \revnot{S_i} > p > \revnot{B-S_i}. \] Among the voters in
  $V''$ we have:
  $2$ voters with preference order $p > d > B$,
  $n-k$ voters with preference order $p > B > d$, and
  $n$ voters with preference order $B > p > d$.
  We claim that it is possible to ensure that $p$ is a unique winner of
  election $E$ via deleting at most $k$ voters if and only if
  $(B,\calS)$ is a ``yes''-instance of $\xthreec$.

  Via routine calculation we see that candidates in election $E$ have
  the following scores:
  \begin{enumerate}
  \item $\score_E(d) = 2n$ (because $N_E(d,p) = 2n$ and for each $b_i
    \in B$, $N_E(d,b_i) = 2n+2$),
  \item $\score_E(p) = 2n-k+2$ (because $N_E(p,d) = 2n-k+2$ and for
    each $b_i \in B$, $N_E(p,b_i) = 2n-k+2$), and
  \item for each $b_i \in B$, $\score_E(b_i) \leq 2n-k$ (because
    $N_E(b_i,d) = 2n-k$).
  \end{enumerate}
  Before any voters are deleted, $d$ is the unique winner with $k-2$
  more points than $p$. Via deleting at most $k$ voters it is possible
  to decrease $d$'s score at most by $k$ points. Let $W$ be a
  collection of voters such that $p$ is the unique winner of
  $E'=(C,V-W)$. We partition $W$ into $W' \cup W''$, where $W'$
  contains those members of $W$ that belong to $V'$ and $W''$ contains
  those members of $W$ that belong to $V''$.  We claim that $W''$ is
  empty. For the sake of contradiction let us assume that $W'' \neq
  \emptyset$.  Let $E'' = (C,V-W'')$. Since every voter in $V''$
  prefers $p$ to $d$, we have that $N_{E''}(p,d) = N_E(p,d)-\|W''\|$
  and, as a result, $\score_{E''}(p) \leq \score_E(p)-\|W''\|$. In
  addition, assuming $W''$ is not empty, it is easy to observe that
  $\score_{E''}(d) \geq \score_E(d)-\|W''\|+1$ (the reason for this is
  that deleting any \emph{single} member of $V''$ does not decrease
  $d$'s score). That is, we have that:
  \begin{eqnarray*}
    \score_{E''}(p) &\leq& 2n-k+2 - \|W''\|, \\
    \score_{E''}(d) &\geq& 2n+1 - \|W''\|.
  \end{eqnarray*}
  So in $E''$, $d$ has at least $k-1$ more points than $p$.
  Since $\|W''\| \geq 1$, we can delete at most $k-1$ voters $W'$ from
election $E''$.  But then $p$ will not be a unique winner of $E'$,
which is a contradiction.

  Thus, $W$ contains members of $V'$ only.  Since $d$ is ranked first
  in every vote in $V'$, deleting voters from $W$ decreases $d$'s
  score by exactly $\|W\|$.  Further, deleting voters $W$ certainly
  decreases $p$'s score by at least one point. Thus, after deleting
  voters $W$ we have:
  \begin{enumerate}
  \item $\score_{E'}(d) = 2n - \|W\|$,
  \item\label{item:maximin:dv} $\score_{E'}(p) \leq 2n - k + 2 -1 =
    2n-k+1$.
  \end{enumerate}
  In consequence, the only possibility that $p$ is a unique winner
  after deleting voters $W$ is that $\|W\| = k$ and we have equality
  in item~\ref{item:maximin:dv} above.  It is easy to verify that this
  equality holds if and only if $W$ contains $k$ voters among $v'_1,
  \ldots, v'_n$ that correspond to an exact cover of $B$ via sets from
  $\calS$ (recall that $k \geq 3$). This proves that our reduction is correct, and since the
  reduction is clearly computable in polynomial time, completes the
  proof of $\np$-hardness of constructive maximin-DV control.

  Let us now consider the destructive case. Let $(B,\calS)$ be our
  input $\xthreec$ instance (with $B$ and $\calS$ as in the
  constructive case). We form election $E = (C,V)$ which is identical
  to the one created in the constructive case, except that $V'' =
  (v''_1, \ldots, v''_{2n-k})$ and we set these voters' preference
  orders as follows: There is one voter with preference order $p > d >
  B$, $n-k$ voters with preference order $p > B > d$, and $n-1$ voters
  with preference order $B > p > d$. (That is, compared to the
  constructive case, we remove one voter with preference order $p > d
  > B$ and one with preference order $B > p > d$.) It is easy to see
  that $d$ is the unique winner of election $E$ and we claim that he
  or she can be prevented from being a unique winner via deleting at
  most $k$ voters if and only if there is an exact cover of $B$ by $k$
  sets from $\calS$.

  Via routine calculation, it is easy to verify that $\score_E(d) =
  2n$, and that $\score_E(p) = 2n-k$. The former holds because
  $N_E(d,p) = 2n$ and $N_E(d,b_i) = 2n+1$ and the latter holds because $N_E(p,d) = 2n-k$ and
  for each candidate $b_i \in B$ we have $N_E(p,b_i) = 2n-k+1$. In
  addition, each candidate $b_i \in B$ has score at most
  $2n-k-1$. Thus, it is possible to ensure that $d$ is not a unique
  winner via deleting at most $k$ voters if and only if there are
  exactly $k$ voters deleting whom would decrease the score of $d$ by
  $k$ points and would not decrease $p$'s score. Let us assume that
  such a collection of voters exists and let $W$ be such a collection.  Since every
  voter in $V''$ prefers $p$ to $d$, clearly $W$ does not contain any
  voter in $V''$. Thus, $W$ contains exactly $k$ voters from
  $V'$. Since for each $b_i \in B$ we have $N_E(p,b_i) = 2n-k+1$, for
  each $b_i \in B$ $W$ contains at most one voter who prefers $p$ to
  $b_i$. Since $\|B\| = 3k$ and $k \geq 3$, this implies that $W$ contains exactly a
  collection of voters corresponding to some exact cover of $B$ by sets in
  $\calS$.  This completes the proof for the destructive case.
\end{proof}

\subsection{Bribery in Maximin}

We now move on to bribery in maximin. Given the previous results, it
is not surprising that maximin is resistant both to constructive
bribery and to destructive bribery.  
Our proof is an
application of 
the ``UV technique'' 
of Faliszewski et al.~\cite{fal-hem-hem-rot:j:llull}.
Very informally, the idea is to build an
election in a way that ensures that the briber is limited to bribing
only those voters who rank two special candidates ahead of the
preferred one.
\begin{theorem}
  Maximin is resistant to constructive and destructive BV control.
\end{theorem}

\begin{proof}
  Our proofs follow via reductions from $\xthreec$.
  The reduction for the constructive case is almost identical the one
  for the constructive case and thus we will consider both cases in
  parallel.

  Our reductions work as follows. Let $(B,\calS)$ be an instance of
  $\xthreec$, where $B = \{b_1, \ldots, b_{3k}\}$ is a set of $3k$
  distinct elements, and $\calS = \{S_1, \ldots, S_n\}$ is a family of
  $3$-element subsets of $B$. (Without loss of generality, we assume that
  $n > k > 1$.  If this is not the case, 
  it is trivial to verify if $(B,\calS)$ is a ``yes''
  instance of $\xthreec$.)  We construct a set of candidates $C =
  \{p,d,s\} \cup B$, where $p$ is our preferred
  candidate (the goal in the constructive setting is to ensure $p$ is
  a unique winner) and $d$ is our despised candidate (the goal in the
  destructive setting is to prevent $d$ from being a unique winner).
  We construct six collections of voters, $V^1, V^2, V^3, V^4, V^5,
  V^6$, as follows:
  \begin{enumerate}
  \item $V^1$ contains $2n$ voters, $v^1_1, \ldots, v^1_{2n}$. For
    each $i$, $1 \leq i \leq n$, voters $v^1_i$ and $v^1_{i+n}$ have the
    following preference orders:
    \begin{eqnarray*}
      v^1_i    &:&  d > s > S_i > p > B-S_i \\
      v^1_{n+i} &:&  \revnot{B-S_i} > p > \revnot{S_i} > d > s.
    \end{eqnarray*}
  
  \item $V^2$ contains $2k$ voters, $v^2_1, \ldots, v^2_{2k}$. For
    each $i$, $1 \leq i \leq k$, voters $v^2_i$ and $v^2_{i+k}$ have the
    following preference orders:
    \begin{eqnarray*}
      v^2_i    &:&  s > d > p > B \\
      v^2_{k+i} &:&  \revnot{B} > d > p > s.
    \end{eqnarray*}

  \item $V^3$ contains $2k$ voters, $v^3_1, \ldots, v^3_{2k}$. For
    each $i$, $1 \leq i \leq k$, voters $v^3_i$ and $v^3_{i+k}$ have the
    following preference orders:
    \begin{eqnarray*}
      v^3_i    &:&  d > s > p > B \\
      v^3_{k+i} &:&  \revnot{B} > s > p > d. 
    \end{eqnarray*}

  \item $V^4$ contains $4k$ voters, $v^4_1, \ldots, v^4_{4k}$. For
    each $i$, $1 \leq i \leq 2k$, voters $v^4_i$ and $v^4_{i+2k}$ have the
    following preference orders:
    \begin{eqnarray*}
      v^4_i    &:&  d > B > p > s \\
      v^4_{2k+i} &:& s > p > d > \revnot{B}.
    \end{eqnarray*}

  \item $V^5$ contains $2$ voters, $v^5_1, v^5_{2}$ with the following
    preference orders
    \begin{eqnarray*}
      v^5_1    &:&  s > B > p > d \\
      v^5_{2}  &:&  d > \revnot{B} > p > s.
    \end{eqnarray*}

  \item $V^6$ contains a single voter, $v^6_1$, with preference order
    $p > d > s > B$.
  \end{enumerate}

  We form two elections, $E_c$ and $E_d$, where $E_c =
  (C,V^1\cup\cdots\cup V^6)$ and $E_d = (C,V^1\cup\cdots\cup V^5)$; that is, $E_c$
  and $E_d$ are identical except $E_d$ does not contain the single
  voter from $V^6$.  $E_c$ contains $2n+8k+3$ voters and $E_d$
  contains $2n+8k+2$ voters. Values of $N_{E_c}$ and $N_{E_d}$ for
  each pair of candidates are given in Table~\ref{tab:bribery}.

\begin{table}[!tbp]
\begin{center}
\subfigure[Values of $N_{E_c}(\cdot,\cdot)$.]{%
  \begin{tabular}{c|c|c|c|c}
   \hline
        & $p$ & $d$ & $s$ & $B$ \\
   \hline
   $p$  & --       & $n+3k+2$ & $n+3k+2$  & $n+4k+1$  \\
   $d$  & $n+5k+1$ & --       & $2n+4k+2$ & $n+6k+2$  \\
   $s$  & $n+5k+1$ & $4k+1$   & --        & $n+4k+2$  \\
   $B$  & $n+4k+2$ & $n+2k+1$ & $n+4k+1$  & $\leq n+4k+2$  \\
   \hline
  \end{tabular}
}
\subfigure[Values of $N_{E_d}(\cdot,\cdot)$.]{%
  \begin{tabular}{c|c|c|c|c}
   \hline
        & $p$ & $d$ & $s$ & $B$ \\
   \hline
   $p$  & --       & $n+3k+1$ & $n+3k+1$  & $n+4k$  \\
   $d$  & $n+5k+1$ & --       & $2n+4k+1$ & $n+6k+1$  \\
   $s$  & $n+5k+1$ & $4k+1$   & --        & $n+4k+1$  \\
   $B$  & $n+4k+2$ & $n+2k+1$ & $n+4k+1$  & $n+4k+1$  \\
   \hline
  \end{tabular}
}
\end{center}
\caption{\label{tab:bribery}Values of $N_{E_c}(\cdot,\cdot)$ and $N_{E_d}(\cdot,\cdot)$ 
  for each pair of candidates. Let $E$ be one of $E_c, E_d$. An entry
  in row $c' \in \{p,d,s\}$ and column $c'' \in \{p,d,s\}$, $c' \neq c''$, of the appropriate table
  above gives value $N_E(c',c'')$. For row $B$ and for column $B$ we adopt
  the following convention.
  For each $c \in \{p,d,s\}$ and for each $b_i \in B$, an entry in row $B$ and column $c$ is
  equal to $N_E(b_i,c)$.
  For each $c \in \{p,d,s\}$ and for each $b_i \in B$, an entry in row $c$ and column $B$ is
  equal to $N_E(c,b_i)$.
  For each two distinct $b_i,b_j \in B$, the entry in row $B$ and column $B$ is the upper
  bound on $N_E(b_i,b_j)$. (For $E_d$ this entry is, in fact, exact.)}
\end{table}

For the constructive case, we claim that it is possible to ensure that
$p$ is a unique winner of election $E_c$ by bribing at most $k$ voters if and
only if $(B,\calS)$ is a ``yes'' instance of $\xthreec$. Let us now
prove this claim.  By inspecting Table~\ref{tab:bribery},
and recalling that $n > k > 1$, we see that
$\score_{E_c}(p) = n+3k+2$, $\score_{E_c}(d) = n+5k+1$,
$score_{E_c}(s) = 4k+1$, and for each $b_i \in B$,
$\score_{E_c}(b_i) \leq n+2k+1$. That is, prior to any bribing, $d$ is
the unique winner and $p$ has the second highest score.

It is easy to see that by bribing $t \leq k$ voters, the briber
can change each candidate's score by at most $t$ points. Thus, for
the bribery to be successful, the briber has to bribe exactly $k$
voters in such a way that $d$'s score decreases to $n+4k+1$ and $p$'s
score increases to $n+4k+2$. To achieve this, the briber has to
find a collection $V'$ of voters such that $\|V'\| = k$, and
\begin{enumerate}
  \item each voter in $V'$ ranks $p$ below both $d$ and $s$, and
  \item for each $b_i \in B$, there is a voter in $V'$ who ranks $p$
    below $b_i$.
\end{enumerate}
The only voters that satisfy the first condition are $v^1_1,
\ldots, v^1_n, v^2_1,\ldots, v^2_k, v^3_1,\ldots, v^3_k$.  Further,
among these voters only $v^1_1, \ldots, v^1_n$ rank $p$ below some
member of $B$ and, in fact, for each $i$, $1 \leq i \leq n$, $v^1_i$
ranks $p$ below exactly three members of $B$. Thus, it is easy to see
that each $k$ voters from $v^1_1, \ldots, v^1_n, v^2_1,\ldots, v^2_k,
v^3_1,\ldots, v^3_k$ that satisfy the second condition correspond
naturally to a cover of $B$ by sets from $\calS$.  (Note that it
suffices that the briber bribes voters in $V'$ to rank $p$ first
without changing the votes in any other way, and that changing the
votes in any other way than ranking $p$ first is not necessary.)  As a
result, if it is possible to ensure that $p$ is a winner of $E_c$ by
bribing at most $k$ voters then $(B,\calS)$ is a ``yes'' instance of
$\xthreec$. 
For the other direction, it is easy to verify that if $(B,\calS)$ is a
``yes'' instance of $\xthreec$ then bribing $k$ voters from $v^1_1,
\ldots, v^1_n$ that correspond to a cover of $B$ to rank $p$ first
suffices to ensure that $p$ is a unique winner.  This completes the
proof for the constructive case.

For the destructive case, we claim that it is possible to ensure that
$d$ is not a unique winner of $E_d$ if and only if $(B,\calS)$ is a
``yes'' instance of $\xthreec$. The proof is analogous to the
constructive case: It suffices to note that $p$ is the only candidate
that can possibly tie for victory with $d$.
The rest of the proof
proceeds as for the constructive case.
\end{proof}

\subsection{Connection to Dodgson Voting}
\label{sec:dodgson}
We conclude our discussion of (control in) maximin voting with a small
detour, showing a connection between maximin and the famous voting
rule (i.e., election system) of Dodgson.

Dodgson voting,
proposed in the 19th century by
Charles Lutwidge 
Dodgson,\footnote{Dodgson is better known as Lewis Carroll, the renowned
author
  of ``Alice's Adventures in
  Wonderland.''}
works as
follows~\cite{dod:unpubMAYBE:dodgson-voting-system}. 
Let $E = (C,V)$ be an election, where $C = \{c_1, \ldots,
c_m\}$ and $V = (v_1, \ldots, v_n)$. For a candidate $c_i \in C$, the
Dodgson score of $c_i$, denoted $\score_E^D(c_i)$, is the smallest
number of sequential swaps of adjacent candidates on the preference lists of
voters in $V$ needed to make $c_i$ become the Condorcet winner. The
candidates with the lowest score are the Dodgson
election's winners.  That is, Dodgson defined his system to 
elect those candidates that are closest to being Condorcet winners in
the sense of adjacent-swaps distance.
Although Dodgson's eighteenth-century election system was 
directly defined in terms of distance, there remains ongoing interest
in understanding the classes of voting rules that can be captured in
various distance-based frameworks (see,
e.g.,~\cite{mes-nur:b:distance-realizability,elk-fal-sli:c:distance-rational}).

Unfortunately, it is known that deciding whether a given candidate is
a winner according to Dodgson's rule is quite complex.  In fact, 
Hemaspaandra, Hemaspaandra, and Rothe~\cite{hem-hem-rot:j:dodgson},
strengthening an NP-hardness result of Bartholdi, Tovey,
and Trick~\cite{bar-tov-tri:j:who-won}, showed 
that this problem is 
complete for parallelized
access to NP\@.  That is, it is 
complete for
the $\thetatwo$ level of the polynomial hierarchy.
Nonetheless,
many researchers have sought efficient ways of computing
Dodgson winners, for example by using 
frequently
correct
heuristics~\cite{hem-hom:j:dodgson-greedy,mcc-pri-sli:j:dodgson},
fixed-parameter
tractability (see~\cite{bar-tov-tri:j:who-won,fal-hem-hem:j:bribery,bet-guo-nie:j:dodgson-parametrized} 
and the discussion in 
Footnote~17 of~\cite{fal-hem-hem-rot:j:llull}),
and approximation algorithms for Dodgson
scores~\cite{car-cov-fel-hom-kak-kar-pro-ros:c:dodgson}.

In addition to its high computational cost in determining winners,
Dodgson's rule is often criticized for not having basic
properties one would expect a good voting rule to have.
For example, 
Dodgson's rule 
is not ``weakCondorcet consistent''
(equivalently, it does not 
satisfy Fishburn's ``strict 
Condorcet principle'')~\cite{bra-bri-hem-hem:c:sp2}
and doesn't satisfy homogeneity and monotonicity
(see~\cite{bra:j:dodgson-remarks}, which surveys a number of 
defects of Dodgson's rule).  We provide definitions for the 
latter two notions, as they will be relevant to this section.

\begin{description}
\item[Homogeneity.] We say that a voting rule $\calR$ is homogeneous if
  for each election $E = (C,V)$, where $C = \{c_1, \ldots, c_m\}$ and
  $V = (v_1, \ldots, v_n)$, it holds that $\calR$ has the same winner
  set on $E$ as on $E' = (C,V')$, where $V' = (v_1, v_1, v_2, v_2,
  \ldots, v_n,v_n)$. 

\item[Monotonicity.] We say that a voting rule $\calR$ is monotone if
  for each election $E = (C,V)$, where $C = \{c_1, \ldots, c_m\}$ and
  $V = (v_1, \ldots, v_n)$, it holds that if some candidate $c_i \in C$
  is a winner of $E$ then $c_i$ is also a winner of an election $E'$ that
  is identical to $E$ except that some voters rank $c_i$ higher (without
  changing the relative order of all the remaining candidates).
\end{description}

Continuing the 
Caragiannis et al.~\cite{car-cov-fel-hom-kak-kar-pro-ros:c:dodgson}
line of research on approximately computing Dodgson
scores, 
Caragiannis et
al.~\cite{car-kak-kar-pro:c:dodgson-acceptable} devised an
approximation algorithm for computing Dodgson scores that, given an
election $E = (C,V)$, where $C = \{c_1, \ldots, c_m\}$ and $V = (v_1,
\ldots, v_n)$ and a candidate $c_i$ in $C$, computes in polynomial
time a nonnegative integer $\scd_E(c_i)$ such that $\score_E^D(c_i)
\leq \scd_E(c_i)$ and 
$\scd_E(c_i) =  O(m\log m)\cdot \score_E^D(c_i)$.  That is, the
algorithm given by Caragiannis et
al.~\cite{car-kak-kar-pro:c:dodgson-acceptable} is, in 
a natural sense, an $O(m\log
m)$-approximation of the 
Dodgson score.\footnote{Throughout
this section, we use the notion ``$f(m)$-approximation of
$g$'' in the sense it is typically
used when dealing with minimization problems.
That is, we mean that the approximation outputs a value that is
at least $g$ and at most $f(m)\cdot g$.  
We slightly abuse the 
interaction between this notation and Big-Oh notation, in 
the quite standard and intuitive way.
And we assume that the argument domain that $g$ and 
the approximation share is clear from context---in this 
paper, their arguments are an election $E$ and 
a candidate $c_i$.}
This algorithm has
additional properties: If one defines a voting rule to elect those
candidates that have lowest scores according to the algorithm, then
that voting rule is Condorcet consistent (i.e., 
when 
a Condorcet winner exists, he or she 
is the one and only winner under the voting rule), 
homogeneous, and monotone.

The result of Caragiannis et
al.~\cite{car-kak-kar-pro:c:dodgson-acceptable} is very interesting,
but unfortunately the voting rule defined by their approximation
algorithm is somewhat complicated and arguably might seem not to be
very natural.  We now show that the maximin rule---which like
the Caragiannis et al.\ rule is Condorcet-consistent, homogeneous, and
monotone, but which in addition is a long-existing and natural
rule---also elects candidates that are, in a certain
different yet precise sense, ``close'' to being Dodgson
winners.   
Our proof is inspired by that of
Caragiannis et al.~\cite{car-kak-kar-pro:c:dodgson-acceptable}.

\begin{theorem}
  \label{thm:maximin-as-dodgson}
  Let $E = (C,V)$ be an election and let $W \subseteq C$ be
  a set of candidates that win in $E$ according to the maximin rule.
  Let $m = \|C\|$ and let $s = \min_{c_i \in C}\score_E^D(c_i)$.
  For each $c_i \in W$ it holds that $s \leq \score_E^D(c_i) \leq m^2s$.
\end{theorem}
\begin{proof}
Let us fix an election $E = (C,V)$ with $C = \{c_1, \ldots c_m\}$ and
$V = (v_1, \ldots, v_n)$. For each two candidates $c_i,c_j \in C$ we
define $\df_E(c_i,c_j)$
to be the smallest number $k$ such that if $k$
voters in $V$ changed their preference order to rank $c_i$ ahead of
$c_j$, then $c_i$ would be preferred to $c_j$ by more than half of the
voters. Note that if for some $c_i,c_j \in C$ we have $\df_E(c_i,c_j)
> 0$ then
\[
N_E(c_i,c_j) + \df_E(c_i,c_j) = \left\lfloor \frac{n}{2} \right\rfloor + 1.
\]
For each candidate $c_i \in C$ we define $\scd'_E(c_i)$ to be
\[
  \scd'_E(c_i) = m^2\max\{ \df_E(c_i,c_j) \mid c_j \in C - \{c_i\}\}.
\]

We now prove that 
$\scd'$ is an $m^2$-approximation of the Dodgson score.

\begin{lemma}\label{thm:m2-sc}
For each $c_i \in C$ it holds that $ \score^D_E(c_i)
\leq \scd'_E(c_i) \leq m^2\score^D_E(c_i)$.
\end{lemma}
\begin{proof}
  Let us fix some $c_i \in C$.  To see that the second inequality in
  the lemma statement holds, note that $\max\{ \df_E(c_i,c_j) \mid c_j
  \in C - \{c_i\}\} \leq \sum _{c_j \in C-\{c_i\}}\df_E(c_i,c_j) \leq
  \score_E^D(c_i)$ because for each candidate $c_k$ we, at least, have
  to perform $\df_E(c_i,c_k)$ swaps to ensure that $c_i$ defeats $c_k$
  in their majority head-to-head contest. Thus, after multiplying by
  $m^2$, we have
  \[
  m^2 \max\{ \df_E(c_i,c_j) \mid c_j \in C - \{c_i\}\} \leq
  m^2\score_E^D(c_i).
\]
  Let us now consider the first inequality. Let $c_k$ be some candidate
  in $C - \{c_i\}$. To make sure that $c_i$ is ranked higher than $c_k$
  by more than half of the voters, we can shift $c_i$ to the first
  position in the preference lists of $\max\{ \df_E(c_i,c_j) \mid c_j
  \in C - \{c_i\}\} \geq \df_E(c_i,c_k)$ voters (or, all the remaining
  voters if less than $\max\{ \df_E(c_i,c_j) \mid c_j \in C - \{c_i\}\}$
  voters do not rank $c_i$ as their top choice). This requires at most
  $m$ adjacent swaps per voter.  Since there are $m-1$ candidates in $C
  - \{c_i\}$, $m^2\max\{ \df_E(c_i,c_j) \mid c_j \in C - \{c_i\}\}$
  adjacent swaps are certainly sufficient to make $c_i$ a Condorcet
  winner. 
  \renewcommand{\qed}{\hspace{\stretch{1}}(Lemma~\ref{thm:m2-sc})~$\qedsymbol$}
\end{proof}

It remains to show that if some candidate $c_i$ is a maximin winner
in $E$ then $\scd'_E(c_i)$ is minimal.
Fortunately, this is easy to
see. If some candidate $c_i$ is a Condorcet winner of $E$ then he or
she is the unique maximin winner and he or she is the unique candidate
$c_i$ with $\scd'_E(c_i) = 0$. Let us assume that there is no
Condorcet winner of $E$. Let us fix some candidate $c_i \in C$ and let
$c_k \in C - \{c_i\}$ be a candidate such that $\scd'_E(c_i) = m^2
\df_E(c_i,c_k)$. That is, $\df_E(c_i,c_k) = \max\{ \df_E(c_i,c_j) \mid
c_j \in C - \{c_i\}\}$ and $\df_E(c_i,c_k) > 0$. Due to this last fact
and our choice of $c_k$, we have $\df_E(c_i,c_k) = \left\lfloor
  \frac{n}{2}\right\rfloor+1 - N_E(c_i,c_k)$ and so
\[ 
   N_E(c_i,c_k) = \left\lfloor\frac{n}{2}\right\rfloor +1 - \df_E(c_i,c_k) =  \min_{c_j \in C-\{c_i\}}N_E(c_i,c_j) = \score_E(c_i),
\]
where $\score_E(c_i)$ is the maximin score of $c_i$ in $E$.
Thus each candidate $c_i$ with the lowest value $\scd'_E(c_i)$ also
has the highest maximin score. 
\end{proof}

Theorem~\ref{thm:maximin-as-dodgson} says that every maximin winner's
Dodgson score is no less than the Dodgson score of the Dodgson winner(s)
(that fact of course holds trivially), and is no more than $m^2$ times
the Dodgson score of the Dodgson winner(s).  That is, we have proven
that \emph{no candidate whose Dodgson score is more than $m^2$ times
  that of the Dodgson winner(s) can be a maximin winner}.  

Since maximin is Condorcet consistent, homogeneous, and monotone, our
result interestingly relates to the approximation of Caragiannis et
al.~\cite{car-kak-kar-pro:c:dodgson-acceptable}, who achieved an $O(m
\log m)$-approximation of Dodgson score while maintaining Condorcet
consistency, homogeneity and monotonicity (recall the discussion
before Theorem~\ref{thm:maximin-as-dodgson}).  Admittedly, our
``closeness'' factor is $m^2$, which is worse than achieving 
$O(m \log
m)$.  And our closeness is in a different sense, since our theorem is
applying its bound just between Dodgson scores, and just on the winner
set.  In contrast, Caragiannis et
al.~\cite{car-kak-kar-pro:c:dodgson-acceptable} and even our own
Lemma~\ref{thm:m2-sc} 
relate the Dodgson score to the Caragiannis et
al.\ score and the $\scd'$ score, and those
approximations hold for
all candidates.  However, we achieve our $m^2$ closeness factor for
a voting rule, maximin, that is well known and natural.

\section{Fixed-Parameter Tractability}
\label{sec:fpt}

In this section we consider the parameterized complexity of multipronged
control, in particular, the case where we can assume that the number
of candidates is a small constant.  Elections with few candidates are
very natural: For example, in many countries presidential elections
involve only a handful of candidates. The reader can easily 
imagine many other
examples.

The main result of this section is that for many natural election
systems $\electionsystem$ (formally, for all election systems whose
winner determination problem can be expressed via an integer linear
program of a certain form), it holds that the
$\electionsystem$-AC+DC+AV+BV+BV control problem is fixed-parameter
tractable (is in the complexity class FPT) 
for the parameter ``number of candidates,'' both in the
constructive setting
and in the destructive setting. This result combines and
significantly enhances FPT results from the literature, in particular,
from the 
papers~\cite{fal-hem-hem:j:bribery,fal-hem-hem-rot:j:llull}, which 
are the model for and inspiration of this section.  We also
make explicit an ``automatic'' path to such results 
that is 
implicit 
in~\cite{fal-hem-hem:j:bribery,fal-hem-hem-rot:j:llull}.  This 
path should be helpful in letting many future analyses 
be done as tool-application exercises, rather than being case-by-case 
challenges.

In this section we focus exclusively on the number of
candidates as our parameter. 
That is, our parameter is the number of candidates
initially in the election
plus the number 
of candidates (if any) in the set of potential additional candidates.
That is, in terms of the variables we have been using to 
describe multiprong control
the parameter is $\|C\| + \|A\|$.

We mention that researchers sometimes
analyze other parameterizations.
For example, Liu et al.~\cite{fen-liu-lua-zhu:j:parameterized-control},
Liu and Zhu~\cite{liu-zhu:j:maximin},
and Betzler and
Uhlmann~\cite{bet-uhl:j:parameterized-complecity-candidate-control}
consider as the parameter the amount of change that one is allowed to use
(e.g., the number of candidates one can add), 
Bartholdi, Tovey,
and Trick~\cite{bar-tov-tri:j:who-won}, 
Betzler and
Uhlmann~\cite{bet-uhl:j:parameterized-complecity-candidate-control},
and
Faliszewski et al.~\cite{fal-hem-hem-rot:j:llull}
study as the parameter 
the number of
voters (and also sometimes the number of candidates).  
And other parameters
are sometimes used when considering the
so-called possible winner problem, see,
e.g.,~\cite{bet-dor:c:possible-winner-dichotomy,bet-hem-nie:c:parameterized-possible-winner}.
However, we view the parameter ``number of candidates'' as the most
essential and the most natural one. We now proceed with our
discussion of fixed-parameter tractability, with the number of
candidates as the parameter.

Let us consider an election system $\electionsystem$ and a set $C =
\{c_1, \ldots, c_m\}$ of candidates.  There are exactly $m!$
preference orders over the candidates in $C$ and we will refer to them
as $o_1, \ldots, o_{m!}$.  Let us assume that $\electionsystem$ is
anonymous (i.e., the winners of each $\electionsystem$ election do not
depend on the order of votes or the names of the voters, but
only---for each preference order $o_i$---on the number of votes with
that preference order).  We define predicate
$\win_\electionsystem(c_j, n_1, \ldots, n_{m!})$ to be true if and
only if $c_i$ is a unique winner of $\electionsystem$ elections
with $C = \{c_1, \ldots, c_m\}$, where for each $i$, $1 \leq i \leq
m!$, there are exactly $n_i$ voters with preference order $o_i$.
For the rest of this section,
our inequalities always use
one of the four 
operators 
``$> $,''
``$\geq $,''
``$< $,'' and
``$\leq $.''\footnote
{We allow both strict and nonstrict
    inequalities. Since we allow only integer solutions, it is easy to
    simulate strict inequalities with nonstrict ones and to 
    simulate nonstrict inequalities with strict ones, in both 
    cases simply by adding a~``$1$'' to the appropriate site of 
    the inequality.  So we could equally well have allowed just 
    strict, or just nonstrict, inequalities.}

\begin{definition}\label{def:ilp-imp}
  We say that an anonymous election system $\electionsystem$ is
  unique-winner (nonunique-winner) integer-linear-program
  implementable if for each set of candidates $C = \{c_1, \ldots,
  c_m\}$ and each candidate $c_j \in C$ there exists a set $S$ of
  linear inequalities
with
  variables $n_1, \ldots, n_{m!}$ such that:
  \begin{enumerate}
  \item If the integer assignment $n_1 = \hat{n}_1$, $\ldots$, $n_{m!}
    = \hat{n}_{m!}$ satisfies $S$, then each $\hat{n}_i$ belongs to
    $\naturals$,\footnote{It is easy to to put $m!$ inequalities into
      $S$ enforcing this condition.  And this condition will help us
      make the electoral part of our definition meaningful, i.e., it
      will avoid having 
      problems from the restriction in the final
      part of this definition that lets us avoid discussing
      negative numbers of voters.}

    \item $S$ can be computed 
(i.e., obtained)
in time polynomial in 
$m!$,\footnote{We mention 
in passing that if the $m!$ in this part of the definition
were changed to any other computable function of $m$, e.g.,
$m^{m^{m^m}}$, we would still obtain FPT results, and still 
would have them 
hold even in the strengthened version of FPT in which the $f$ of 
``$f(\mbox{\rm{}parameter})\cdot\mbox{\rm{}Inputsize}^{O(1)}$''
is required to be computable.  However, due to $m!$ being the number 
of preference orders over $m$ candidates, having $S$ be obtainable
in time polynomial in $m!$ will in practice be a particularly common 
case.  

We also mention in passing 
that the FPT-establishing framework in this section
and the results it yields,
similarly to the case in our 
work
mentioned 
earlier~\cite{fal-hem-hem:j:bribery,fal-hem-hem-rot:j:llull}, 
not only will apply in the model where votes are input as 
a list of ballots, one per person, but 
also will hold
in the so-called ``succinct'' 
model (see~\cite{fal-hem-hem:j:bribery,fal-hem-hem-rot:j:llull}),
in which we are given the votes not as individual ballots but 
as binary numbers providing the number of voters having each 
preference order 
(or having each occurring preference order).
}
and

    \item for each $(\hat{n}_1,\ldots,\hat{n}_{m!})\in \naturals^{m!}$,
      we have that~(a) holds if and only if~(b) holds, where~(a)
      and~(b) are as follows:
\begin{enumerate}
\item[(a)] $S$ is satisfied by the assignment
$n_1 = \hat{n}_1$, $\ldots$, 
$n_{m!} = \hat{n}_{m!}$.
\item[(b)] 
$c_j$ is
      a unique winner (is a winner) of an $\electionsystem$ election
      in which for each $i$, $1 \leq i \leq m!$, there are exactly $\hat{n}_i$
      voters with preference order $o_i$, where $o_i$ is the $i$'th
      preference order over the set $C$.
\end{enumerate}
  \end{enumerate}
\end{definition}

In a slight abuse of notation, for integer-linear-program
implementable election systems $\electionsystem$ we will simply refer
to the set $S$ of linear inequalities from
Definition~\ref{def:ilp-imp} as $\win_\electionsystem(c_j,n_1, \ldots,
n_{m!})$.  The particular set of candidates will always be clear from
context.  Naturally, it is easy to adapt Definition~\ref{def:ilp-imp}
to apply to approval voting, but for the sake of brevity we will not
do so.   

We are not aware of any natural systems that are
integer-linear-program unique-winner implementable yet not
integer-linear-program nonunique-winner implementable, or vice
versa.  
In this paper we focus on the unique winner model so the reader may
wonder why we defined the nonunique winner variant of
integer-linear-program implementability.  The answer is that, as we
will see later in this section, it is a useful notion when dealing
with destructive control.

The class of election systems that are integer-linear-program
implementable is remarkably broad.  For example, it is 
variously implicit in or 
a consequence of the
results of~\cite{fal-hem-hem:j:bribery} that plurality, veto, Borda,
Dodgson, and each polynomial-time computable (in the number of candidates)
family of  scoring protocols
are integer-linear-program
implementable.\footnote{Let $m$ be the number of candidates.
  A scoring protocol is a vector of $m$ nonnegative
  integers satisfying $\alpha_1 \geq \alpha_2 \geq \cdots \geq \alpha_m$.  
Each candidate receives
  $\alpha_i$ points for each vote that ranks him or her in the $i$'th
  position, and the candidate(s) with most points win. Many election
  systems can be viewed as families of scoring protocols. For example,
  plurality is defined by scoring protocols of the form $(1,0, \ldots,
  0)$, veto is defined by scoring protocols of the form $(1, \ldots,
  1,0)$, and Borda is defined by scoring protocols of the form $(m-1,
  m-2, \ldots, 0)$, where $m$ is the number of candidates.}
For many other election systems (e.g.,
Kemeny~\cite{kem:j:no-numbers,lev-you:j:condorcet} and Copeland) it is
not clear whether they are integer-linear-program implementable, but
there are similar approaches that will be as useful for us.  We will
return to this issue at the end of this section.

\begin{theorem}\label{thm:fpt}
  Let $\electionsystem$ be an integer-linear-program unique-winner
  implementable election system. For 
number of candidates 
as the parameter, 
  constructive $\electionsystem$-AC+DC+AV+DV+BV is in FPT.
\end{theorem}
\begin{proof}
  Let $(C,A,V,W,p, k_\ac,k_\dc, k_\av, k_\dv, k_\bv)$ be our input
  instance of the constructive $\electionsystem$-AC+DC+AV+DV+BV control
  problem, as described in Definition~\ref{def:multiprong}. Let $C =
  \{p,c_1, \ldots, c_{m'}\}$ and $A = \{a_1, \ldots, a_{m''}\}$.  Our
  parameter, the total number of candidates, is $m = m' + m'' + 1$.
  For each subset $K$ of $C \cup A$ we let $o^K_1, \ldots,
  o^K_{\|K\|!}$ mean the $\|K\|!$ preference orders
  over $K$.

  The idea of our algorithm is to perform an exhaustive search through
  all the subsets of candidates $K$, $K \subseteq C \cup A$, and for
  each $K$ check whether (a) it is possible to obtain $K$ from $C$ by
  deleting at most $k_\dc$ candidates and adding at most $k_\ac$
  candidates from $A$, and (b) it is possible to ensure that $p$ is a
  unique winner of election $(K,V)$ by deleting at most $k_\dv$ voters,
  adding at most $k_\av$ voters from $W$, and bribing at most $k_\bv$
  voters. Given $K$, step (a) can easily be implemented in polynomial
  time. To implement step (b), we introduce a linear integer program
  $P(K)$, which is satisfiable if and only if step (b) holds.  Let us
  now fix $K \subseteq C \cup A$ and describe the integer linear
  program $P(K)$. 

  We assume that $p \in K$ as it is not legal to delete $p$ (and it
  would be pointless, given that we want to ensure his or her
  victory). We interpret preference orders of voters in $V$ and
  $W$ as limited to the candidate set $K$.  We use the following
  \emph{constants} in our program. For each $i$, $1 \leq i \leq
  \|K\|!$, we let $n^V_i$ be the number of voters in $V$ with
  preference order $o^K_i$, and we let $n^W_i$ be the number of voters
  in $W$ with preference order $o^K_i$. $P(K)$ contains the following
  variables (described together with their intended interpretation):
  \begin{description}
  \item[Variables $\boldsymbol{av_1, \ldots, av_{\|K\|!}}$.] For each $i$,
    $1 \leq i \leq \|K\|!$, we interpret $av_i$ as the number of
    voters with preference $o^K_i$ that we add from $W$.
  \item[Variables $\boldsymbol{dv_1, \ldots, dv_{\|K\|!}}$.] For each $i$,
    $1 \leq i \leq \|K\|!$, we interpret $dv_i$ as the number of
    voters with preference $o^K_i$ that we delete from $V$.
  \item[Variables $\boldsymbol{bv_{1,1}, bv_{1,2},\ldots, bv_{1,\|K\|!},
      bv_{2,1}, \ldots, bv_{\|K\|!,\|K\|!}}$.] For each $i,j$, $1 \leq
    i,j \leq \|K\|!$, we interpret $bv_{i,j}$ as the number of voters
    with preference $o^K_i$ that, in case $i \neq j$, we bribe to
    switch to preference order $o^K_j$, or, in case $i = j$, we leave
    unbribed.
  \end{description}
  $P(K)$ contains the following constraints.
  \begin{enumerate}

  \item All the variables have nonnegative values.

  \item For each variable $av_i$, $1 \leq i \leq \|K\|!$, there are
    enough voters in $W$ with preference order $o^K_i$ to be
    added. That is, for each $i$, $1 \leq i \leq \|K\|!$, we have a
    constraint $av_i \leq n^W_i$.  Altogether, we can add at
    most $k_\av$ voters so we have a constraint
    $\sum_{i=1}^{\|K\|!}av_i \leq k_\av$.

  \item For each variable $dv_i$, $1 \leq i \leq \|K\|!$, there are
    enough voters in $V$ with preference order $o^K_i$ to be
    deleted. That is, for each $i$, $1 \leq i \leq \|K\|!$, we have a
    constraint $dv_i \leq n^V_i$.  Altogether, we can delete at most
    $k_\dv$ voters so we have a constraint $\sum_{i=1}^{\|K\|!}dv_i
    \leq k_\dv$.

  \item For each variable $bv_{i,j}$, $1 \leq i,j \leq \|K\|!$, there
    are enough voters with preference $o^K_i$ to be bribed. That is,
    for each $i$, $1 \leq i \leq \|K\|!$, we have a constraint
    $\sum_{j=1}^{\|K\|!}bv_{i,j} = n^V_i+av_i-dv_i$ (the equality
    comes from the fact that for each $i$, $1 \leq i \leq \|K\|!$,
    $bv_{i,i}$ is the number of voters with preference $o^K_i$ that we
    do not bribe). Altogether, we can bribe at most $k_\bv$ voters so
    we also have a constraint \[
    \left( \sum_{i=1}^{\|K\|!}\sum_{j=1}^{\|K\|!}bv_{i,j} \right) - \sum_{i=1}^{\|K\|!}bv_{i,i} \leq k_\bv.\]

  \item Candidate $p$ is the unique winner of the election after we
    have executed all the adding, deleting, and bribing of voters. Using
    the fact that $\electionsystem$ is integer-linear-program
    unique-winner implementable, we can express this as
    $\win_\electionsystem(p, \ell_1, \ldots, \ell_{\|K\|!})$, where we
    substitute each $\ell_j$, $1 \leq j \leq \|K\|!$, by
    $\sum_{i=1}^{\|K\|!}bv_{i,j}$ (note that, by previous
    constraints, variables describing bribery already take into
    account adding and deleting voters). This is a legal
    integer-linear-program constraint as $\win_\electionsystem(p,
    \ell_1, \ldots, \ell_{\|K\|!})$ is simply a conjunction of linear
    inequalities over $\ell_1, \ldots, \ell_{\|K\|!}$.
  \end{enumerate}

The number of variables and the number of
  inequalities in $P(K)$ are each polynomially bounded in $m!$.
Keeping in mind
Definitions~\ref{def:multiprong} and~\ref{def:ilp-imp}, it is
  easy to see that program $P(K)$ does exactly what we expect it
  to. 
And
testing whether $P(K)$ is satisfiable (i.e., has an 
integer solution, as we are in the framework of an 
integer linear program) is in FPT,
with respect to the number of candidates being our 
parametrization,
by using Lenstra's algorithm~\cite{len:j:integer-fixed}.
 Thus
  our complete FPT algorithm for the $\electionsystem$-AC+DC+AV+DV+BV
  problem works as follows. For each subset $K$ of $C \cup A$ that
  includes $p$ we execute the following two steps:
  \begin{enumerate}
  \item Check whether it is possible to obtain $K$ from $C$ by deleting at
    most $k_\dc$ candidates and by adding at most $k_\ac$ candidates from
    $A$.
  \item Form linear program $P(K)$ and check whether it has any integral
    solutions using the algorithm of
    Lenstra~\cite{len:j:integer-fixed}. Accept if so.
  \end{enumerate}
  If after trying all sets $K$ we have not accepted, then reject.

  {}From the previous discussion, this algorithm is correct. Also, since (a)
  there are exactly $2^{m-1}$ sets $K$ to try, (b) executing the first
  step above can be done in time polynomial in $m$, and (c) the second
  step is in FPT (given that $m$ is the parameter), constructive
  $\electionsystem$-AC+DC+AV+DV+BV is in FPT for parameter $m$.
\end{proof}

The above theorem deals with constructive control only. However, using
its proof, it is easy to prove a destructive variant of the result.
We say that an election system is 
strongly voiced~\cite{hem-hem-rot:j:destructive-control}
if it holds that
whenever there is at least one candidate, there is at least one
winner.

\begin{cor}\label{cor:fpt-destructive}
  Let $\electionsystem$ be a strongly voiced, integer-linear-program
  nonunique-winner implementable election system. Destructive
  $\electionsystem$-AC+DC+AV+DV+BV is in FPT for the parameter number
  of candidates.
\end{cor}

To see that the corollary holds, it is enough to note that for
strongly voiced election systems a candidate can be prevented from
being a unique winner if and only if some other candidate can be
made a (possibly nonunique) winner (see e.g., Footnote~5
of~\cite{hem-hem-rot:j:destructive-control} for a relevant
discussion). Thus to prove Corollary~\ref{cor:fpt-destructive}, we
can simply use an algorithm that for each candidate other than the
despised one sees whether that candidate can be made a 
(perhaps nonunique)
winner, and if any can be made a (perhaps nonunique) winner,
declares destructive control achievable.
(And the precise integer linear programming feasibility problem 
solution given by Lenstra's algorithm will 
reveal
what action achieves the control.)
This can be done in FPT using the algorithm from the proof of
Theorem~\ref{thm:fpt}, adapted to work for the nonunique-winner
problem (this is trivial given that
Corollary~\ref{cor:fpt-destructive} assumes that $\electionsystem$ is
integer-linear-program nonunique-winner implementable).

Let us now go back to the issue that some election systems may not be
integer-linear-program implementable.  As an
example, let us consider maximin.  Let $E = (C,V)$ be an election,
where $C = \{c_1, \ldots, c_m\}$ and $V = (v_1, \ldots, v_n)$. As
before, by $o_1, \ldots, o_{m!}$ we mean the $m!$ possible preference
orders over $C$, and for each $i$, $1 \leq i \leq m!$, by $n_i$ we
mean the number of voters in $V$ that report preference order
$o_i$. For each $c_i$ and $c_j$ in $C$, $c_i \neq c_j$, we let $O(c_i,c_j)$ be
the set of preference orders over $C$ where $c_i$ is preferred to
$c_j$.  Let $k = (k_1, \ldots, k_m)$ be a vector of nonnegative
integers such that for each $i$, $1 \leq i \leq m$, it holds that $1
\leq k_i \leq m$. For such a vector $k$ and a candidate $c_\ell \in C$
we define $M(c_\ell, k_1, \ldots, k_m)$ to be the following set of
linear integer inequalities:
\begin{enumerate}
\item For each candidate $c_i$, his or her maximin score is equal to
  $N_E(c_i,c_{k_i})$. That is, for each $i,j$, $1 \leq i,j \leq m$, $i
  \neq j$, we have constraint
  $\sum_{o_k \in O(c_i,c_{k_i})} n_k \leq \sum_{o_k \in O(c_i,c_j)}n_k$
\item $c_\ell$ has the highest maximin score in election $E$ and thus
  is the unique winner of $E$. That is, for each $i$, $1 \leq i \leq m$,
  $i \neq \ell$, we have constraint
  $\sum_{o_k \in O(c_\ell,c_{k_\ell})} n_k > \sum_{o_k \in O(c_i,c_{k_i})} n_k$.
\end{enumerate} 
It is easy to see that $c_\ell$ is a unique maximin winner of $E$ if
and only if there is a vector $k = (k_1, \ldots, k_m)$ such that all
inequalities of $M(c_\ell,k_1, \ldots, k_m)$ are satisfied. It is also
clear how to modify the above construction to handle the nonunique
winner case.  Since there are only $O(m^m)$ vectors $k$ to try and
each $M(c_\ell, k_1, \ldots, k_m)$ contains $O(m^2)$ inequalities, it
is easy to modify the proof of Theorem~\ref{thm:fpt} to work for
maximin: Assuming that one is interested in ensuring candidate
$c_\ell$'s victory, one simply has to replace program $P(K)$ in the
proof of Theorem~\ref{thm:fpt} with a family of programs that each
include a different $M(c_\ell, k_1, \ldots, k_m)$ for testing if
$c_\ell$ had won. And one would accept if any of these were
satisfiable. Thus we have the following result.

\begin{cor}
  Constructive AC+DC+AV+DV+BV control and destructive AC+DC+AV+DV+BV
  control are both in FPT for maximin for the parameter number of
  candidates.
\end{cor}

The above construction for the winner problem in maximin can be viewed
as, in effect, a disjunction of a set of integer linear programs. Such
constructions for the winner problem have already been obtained for
Kemeny in~\cite{fal-hem-hem:j:bribery} and for
Copeland in~\cite{fal-hem-hem-rot:j:llull}. Thus
we have the following theorem.

\begin{cor}
With number of candidates as the parameter,
 constructive AC+DC+AV+DV+BV
  control and destructive AC+DC+AV+DV+BV control are in FPT for
  Kemeny and, 
for each 
 each rational $\alpha$, 
$0 \leq \alpha \leq 1$, 
 for Copeland$^\alpha$.
\end{cor}

We conclude with an important caveat. The FPT algorithms of this section
are very broad in their coverage, but in practice they 
would be difficult to use as their
running time depends on 
(the fixed-value parameter) $m$ in a very fast-growing way and as Lenstra's
algorithm has a large multiplicative constant in its 
polynomial
running
time. Thus the results of this section should best
be interpreted
as indicating that, for multipronged control in our setting,
it is impossible to prove 
non-FPT-ness 
(and so it clearly is impossible to prove fixed-parameter 
\emph{hardness} in terms of the levels of the 
so-called ``W'' hierarchy of fixed-parameter complexity,
unless that hierarchy collapses to FPT).
If one is interested in
truly practically 
implementing a multipronged control attack, one should probably
devise a
problem-specific
algorithm rather than using our very generally applicable FPT construction.

\section{Conclusions}
\label{sec:conclusion}

We have shown that combining various types of control into multiprong
control attacks is a useful technique. It allows us to study more
realistic control models, to express control vulnerability results and
proofs in a compact way, and to obtain vulnerability results that are
stronger than would be obtained for single prongs alone.

The main finding of our paper is that, to the extent to which we can
draw conclusions from the set of election systems that we have
studied, vulnerabilities to basic control types can often be combined to
form a vulnerability to their multipronged control combination.
(Table~\ref{tab:results} summarizes our results regarding the five
election systems we have focused on in this paper.)
However, we have also seen that there exists a natural election
system that is
vulnerable to both constructive AC control and constructive AV control
but that is resistant to constructive AC+AV control.
We have also
shown that as far as fixed-parameter tractability goes, at least with
respect to the parameter number of candidates, a very broad class of
election systems is vulnerable to the full AC+DC+AV+DV+BV control
attack.   And we have taken a small detour and proven that 
no candidate whose Dodgson score is more than $\|C\|^2$ times the 
Dodgson winner's score can be a maximin winner.

\begin{table}[!tbp]
\begin{center}
\small
\begin{tabular}{|c|c|c|c|c|c|c|c|c|c|c|}
\hline
Control type & \multicolumn{2}{|c|}{plurality} & 
\multicolumn{2}{|c|}{Condorcet} &
\multicolumn{2}{|c|}{Copeland$^{0.5}$} &
\multicolumn{2}{|c|}{approval} &
\multicolumn{2}{|c|}{maximin} \\
\cline{2-11}
& Con. & Des. & Con. & Des.& Con. & Des.& Con. & Des.& Con. & Des.\\
\hline
AC           & R & R & I & V & R & V & I & V & R & V \\
AC$_{\rm u}$ & R & R & I & V & R & V & I & V & V & V \\
DC           & R & R & V & I & R & V & V & I & V & V \\
\hline
AV           & V & V & R & V & R & R & R & V & R & R \\
DV           & V & V & R & V & R & R & R & V & R & R \\
\hline
BV           & V & V & R & V & R & R & R & V & R & R \\
\hline
\end{tabular}
\caption{\label{tab:results} Resistance to basic control types for the
  five main election systems studied in this paper. In the table, I
  means the system is immune to the given control type, R means
  resistance, and V means vulnerability.  As shown in this paper, for
  each of the five election systems, all listed constructive
  vulnerabilities combine and all listed destructive vulnerabilities
  combine. All remaining prongs combine as described by
  Corollary~\ref{combining-prongs}.  Constructive results for AC,
  AC$_{\rm u}$, DC, AV, and DV for plurality and Condorcet are due
  to~\cite{bar-tov-tri:j:control} and their corresponding destructive
  results are due to~\cite{hem-hem-rot:j:destructive-control}.  All
  results for AC ,AC$_{\rm u}$, DC, AV, and DV for approval are due
  to~\cite{hem-hem-rot:j:destructive-control}. All results regarding
  Copeland, are due to~\cite{fal-hem-hem-rot:j:llull}. Constructive
  bribery results for plurality and approval are due
  to~\cite{fal-hem-hem:j:bribery}, and the constructive bribery result
  for Condorcet is implicit in~\cite{fal-hem-hem-rot:j:llull}. All the
  remaining results (i.e., all results regarding maximin, and
  destructive bribery results for plurality, approval, and Condorcet)
  are due to this paper.
}
\end{center}
\end{table}

This paper studies multipronged control where the prongs may include
various standard types of control or bribery. However, it is easy to
see that our framework can be naturally extended to include
manipulation.  To do so, one would have to allow some of the
voters---the manipulators---to have blank preference orders and, if
such voters were to be included in the election, the controlling agent
would have to decide on how to fill them in. It is interesting that in
this model the controlling agent might be able to add manipulative
voters (if there were manipulators among the voters that can be added)
or even choose to delete them (it may seem that deleting manipulators
is never useful but Zuckerman, Procaccia, and
Rosenschein~\cite{pro-ros-zuc:j:borda} give an example where deleting
a manipulator is necessary to make one's favorite candidate a winner
of a Copeland election).

We mention as a natural but 
involved
open direction the study of multipronged control in the setting
where there are multiple controlling agents, each with a different
goal, each controlling a different prong. In such a setting, it is
interesting to consider game-theoretic scenarios as well as
situations in which, for example, one of the controlling agents is
seeking an action that will succeed regardless of the action of the
other attacker.

\section*{Acknowledgments}
Supported in part by NSF grants CCF-0426761, IIS-0713061, and
CCF-0915792, Polish
Ministry of Science and Higher Education grant N-N206-378637,
the Foundation for Polish Science's Homing/Powroty program, AGH University
of Science and Technology grant 11.11.120.865, the ESF's EUROCORES
program LogICCC, and Friedrich Wilhelm Bessel Research Awards to Edith
Hemaspaandra and Lane A. Hemaspaandra.
A preliminary version of this paper appeared in the proceedings of
the 21st International Joint Conference on Artificial Intelligence,
July 2009~\cite{fal-hem-hem:c:multimode}.
We thank Edith Elkind and the anonymous IJCAI referees for helpful comments.

\bibliography{grypiotr2006,gry-multimode-additions}

\newcommand{\etalchar}[1]{$^{#1}$}
\begin{thebibliography}{FHHR09b}

\bibitem[BBHH10]{bra-bri-hem-hem:c:sp2}
F.~Brandt, M.~Brill, E.~Hemaspaandra, and L.~Hemaspaandra.
\newblock Bypassing combinatorial protections: Polynomial-time algorithms for
  single-peaked electorates.
\newblock In {\em Proceedings of the 24th AAAI Conference on Artificial
  Intelligence}. AAAI Press, July 2010.
\newblock To appear.

\bibitem[BD09]{bet-dor:c:possible-winner-dichotomy}
N.~Betzler and B.~Dorn.
\newblock Towards a dichotomy of finding possible winners in elections based on
  scoring rules.
\newblock In {\em Proceedings of the 34th International Symposium on
  Mathematical Foundations of Computer Science}, pages 124--136.
  Springer-Verlag {\it Lecture Notes in Computer Science \#5734}, August 2009.

\bibitem[BEH{\etalchar{+}}]{bau-erd-hem-hem-rot:btoappear:computational-apects%
-of-approval-voting}
D.~Baumeister, G.~Erd\'{e}lyi, E.~Hemaspaandra, L.~Hemaspaandra, and J.~Rothe.
\newblock Computational aspects of approval voting.
\newblock In J.~Laslier and R.~Sanver, editors, {\em Handbook of Approval
  Voting}. Springer.
\newblock To appear.

\bibitem[BGN10]{bet-guo-nie:j:dodgson-parametrized}
N.~Betzler, J.~Guo, and R.~Niedermeier.
\newblock Parameterized computational complexity of {D}odgson and {Y}oung
  elections.
\newblock {\em Information and Computation}, 208(2):165--177, 2010.

\bibitem[BHN09]{bet-hem-nie:c:parameterized-possible-winner}
N.~Betzler, S.~Hemmann, and R.~Niedermeier.
\newblock A multivariate complexity analysis of determining possible winners
  given incomplete votes.
\newblock In {\em Proceedings of the 21st International Joint Conference on
  Artificial Intelligence}, pages 53--58. AAAI Press, July 2009.

\bibitem[Bla58]{bla:b:polsci:committees-elections}
D.~Black.
\newblock {\em The Theory of Committees and Elections}.
\newblock Cambridge University Press, 1958.

\bibitem[BO91]{bar-oli:j:polsci:strategic-voting}
J.~{{Bartholdi}}, III and J.~Orlin.
\newblock Single transferable vote resists strategic voting.
\newblock {\em Social Choice and Welfare}, 8(4):341--354, 1991.

\bibitem[Bra09]{bra:j:dodgson-remarks}
F.~Brandt.
\newblock Some remarks on {D}odgson's voting rule.
\newblock {\em Mathematical Logic Quarterly}, 55(4):460--463, 2009.

\bibitem[BTT89a]{bar-tov-tri:j:manipulating}
J.~{{Bartholdi}}, III, C.~Tovey, and M.~Trick.
\newblock The computational difficulty of manipulating an election.
\newblock {\em Social Choice and Welfare}, 6(3):227--241, 1989.

\bibitem[BTT89b]{bar-tov-tri:j:who-won}
J.~{{Bartholdi}}, III, C.~Tovey, and M.~Trick.
\newblock Voting schemes for which it can be difficult to tell who won the
  election.
\newblock {\em Social Choice and Welfare}, 6(2):157--165, 1989.

\bibitem[BTT92]{bar-tov-tri:j:control}
J.~{{Bartholdi}}, III, C.~Tovey, and M.~Trick.
\newblock How hard is it to control an election?
\newblock {\em Mathematical and Computer Modeling}, 16(8/9):27--40, 1992.

\bibitem[BU09]{bet-uhl:j:parameterized-complecity-candidate-control}
N.~Betzler and J.~Uhlmann.
\newblock Parameterized complexity of candidate control in elections and
  related digraph problems.
\newblock {\em Theoretical Computer Science}, 410(52):43--53, 2009.

\bibitem[CCF{\etalchar{+}}09]{car-cov-fel-hom-kak-kar-pro-ros:c:dodgson}
I.~Caragiannis, J.~Covey, M.~Feldman, C.~Homan, C.~Kaklamanis, N.~Karanikolas,
  A.~Procaccia, and J.~Rosenschein.
\newblock On the approximability of {D}odgson and {Y}oung elections.
\newblock In {\em Proceedings of the 20th Annual ACM-SIAM Symposium on Discrete
  Algorithms}, pages 1058--1067. Society for Industrial and Applied
  Mathematics, January 2009.

\bibitem[CKKP10]{car-kak-kar-pro:c:dodgson-acceptable}
I.~Caragiannis, C.~Kaklamanis, N.~Karanikolas, and A.~Procaccia.
\newblock Socially desirable approximations for {D}odgson's voting rule.
\newblock In {\em Proceedings of the 11th ACM Conference on Electronic
  Commerce}, pages 253--262. ACM Press, June 2010.

\bibitem[CS06]{con-san:c:nonexistence}
V.~Conitzer and T.~Sandholm.
\newblock Nonexistence of voting rules that are usually hard to manipulate.
\newblock In {\em Proceedings of the 21st National Conference on Artificial
  Intelligence}, pages 627--634. AAAI Press, July 2006.

\bibitem[Dod76]{dod:unpubMAYBE:dodgson-voting-system}
C.~Dodgson.
\newblock A method of taking votes on more than two issues.
\newblock Pamphlet printed by the Clarendon Press, Oxford, and headed ``not yet
  published'' (see the discussions
  in~\protect\cite{mcl-urk:b:polsci:classics,bla:b:polsci:committees-elections%
}, both of which reprint this paper), 1876.

\bibitem[DP08]{dob-pro:c:two-voters}
S.~Dobzinski and A.~Procaccia.
\newblock Frequent manipulability of elections: The case of two voters.
\newblock In {\em Proceedings of the 4th International Workshop On Internet And
  Network Economics}, pages 653--664. Springer-Verlag {\it Lecture Notes in
  Computer Science \#5385}, December 2008.

\bibitem[EFS09]{elk-fal-sli:c:distance-rational}
E.~Elkind, P.~Faliszewski, and A.~Slinko.
\newblock On distance rationalizability of some voting rules.
\newblock In {\em Proceedings of the 12th Conference on Theoretical Aspects of
  Rationality and Knowledge}. ACM Press, July 2009.

\bibitem[EFS10]{elk-fal-sli:c:cloning}
E.~Elkind, P.~Faliszewski, and A.~Slinko.
\newblock Cloning in elections.
\newblock In {\em Proceedings of the 24th AAAI Conference on Artificial
  Intelligence}. AAAI Press, July 2010.
\newblock To appear.

\bibitem[ENR09]{erd-now-rot:j:sp-av}
G.~Erd\'{e}lyi, M.~Nowak, and J.~Rothe.
\newblock Sincere-strategy preference-based approval voting fully resists
  constructive control and broadly resists destructive control.
\newblock {\em Mathematical Logic Quarterly}, 55(4):425--443, 2009.

\bibitem[EPR10a]{erd-pir-rot:t:bucklin}
G.~Erd\'{e}lyi, L.~Piras, and J.~Rothe.
\newblock Bucklin voting is broadly resistant to control.
\newblock Technical Report arXiv:1005.4115~[cs.GT], arXiv.org, May 2010.

\bibitem[EPR10b]{erd-pir-rot:t:fallback}
G.~Erd\'{e}lyi, L.~Piras, and J.~Rothe.
\newblock Control complexity in fallback voting.
\newblock Technical Report arXiv:1004.3398~[cs.GT], arXiv.org, April 2010.

\bibitem[FHH]{fal-hem-hem:jtoappear:cacm-survey}
P.~Faliszewski, E.~Hemaspaandra, and L.~Hemaspaandra.
\newblock Using complexity to protect elections.
\newblock {\em Communications of the ACM}.
\newblock To appear.

\bibitem[FHH09a]{fal-hem-hem:j:bribery}
P.~Faliszewski, E.~Hemaspaandra, and L.~Hemaspaandra.
\newblock How hard is bribery in elections?
\newblock {\em Journal of Artificial Intelligence Research}, 35:485--532, 2009.

\bibitem[FHH09b]{fal-hem-hem:c:multimode}
P.~Faliszewski, E.~Hemaspaandra, and L.~Hemaspaandra.
\newblock Multimode attacks on elections.
\newblock In {\em Proceedings of the 21st International Joint Conference on
  Artificial Intelligence}, pages 128--133. AAAI Press, July 2009.

\bibitem[FHHR09a]{fal-hem-hem-rot:j:llull}
P.~Faliszewski, E.~Hemaspaandra, L.~Hemaspaandra, and J.~Rothe.
\newblock Llull and {Copeland} voting computationally resist bribery and
  constructive control.
\newblock {\em Journal of Artificial Intelligence Research}, 35:275--341, 2009.

\bibitem[FHHR09b]{fal-hem-hem-rot:b:richer}
P.~Faliszewski, E.~Hemaspaandra, L.~Hemaspaandra, and J.~Rothe.
\newblock A richer understanding of the complexity of election systems.
\newblock In S.~Ravi and S.~Shukla, editors, {\em Fundamental Problems in
  Computing: {Essays} in Honor of {Professor} {Daniel} {J.} {Rosenkrantz}},
  pages 375--406. Springer, 2009.

\bibitem[FHHR09c]{fal-hem-hem-rot:c:single-peaked-preferences}
P.~Faliszewski, E.~Hemaspaandra, L.~Hemaspaandra, and J.~Rothe.
\newblock The shield that never was: {Societies} with single-peaked preferences
  are more open to manipulation and control.
\newblock In {\em Proceedings of the 12th Conference on Theoretical Aspects of
  Rationality and Knowledge}, pages 118--127. ACM Press, July 2009.

\bibitem[FHS08]{fal-hem-sch:c:copeland-ties-matter}
P.~Faliszewski, E.~Hemaspaandra, and H.~Schnoor.
\newblock Copeland voting: Ties matter.
\newblock In {\em Proceedings of the 7th International Conference on Autonomous
  Agents and Multiagent Systems}, pages 983--990. International Foundation for
  Autonomous Agents and Multiagent Systems, May 2008.

\bibitem[FHS10]{fal-hem-sch:c:copeland01}
P.~Faliszewski, E.~Hemaspaandra, and H.~Schnoor.
\newblock Manipulation of {Copeland} elections.
\newblock In {\em Proceedings of the 9th International Conference on Autonomous
  Agents and Multiagent Systems}, pages 367--374. International Foundation for
  Autonomous Agents and Multiagent Systems, May 2010.

\bibitem[FKN08]{fri-kal-nis:c:quantiative-gib-sat}
E.~Friedgut, G.~Kalai, and N.~Nisan.
\newblock Elections can be manipulated often.
\newblock In {\em Proceedings of the 49th IEEE Symposium on Foundations of
  Computer Science}, pages 243--249. IEEE Computer Society, October 2008.

\bibitem[GJ79]{gar-joh:b:int}
M.~Garey and D.~Johnson.
\newblock {\em Computers and Intractability: {A} Guide to the Theory of
  {NP}-Completeness}.
\newblock {W. H. Freeman and Company}, 1979.

\bibitem[HH09]{hem-hom:j:dodgson-greedy}
C.~Homan and L.~Hemaspaandra.
\newblock Guarantees for the success frequency of an algorithm for finding
  {Dodgson}-election winners.
\newblock {\em Journal of Heuristics}, 15(4):403--423, 2009.

\bibitem[HHR97]{hem-hem-rot:j:dodgson}
E.~Hemaspaandra, L.~Hemaspaandra, and J.~Rothe.
\newblock Exact analysis of {D}odgson elections: {L}ewis {C}arroll's 1876
  voting system is complete for parallel access to {NP}.
\newblock {\em Journal of the ACM}, 44(6):806--825, 1997.

\bibitem[HHR07]{hem-hem-rot:j:destructive-control}
E.~Hemaspaandra, L.~Hemaspaandra, and J.~Rothe.
\newblock Anyone but him: {The} complexity of precluding an alternative.
\newblock {\em Artificial Intelligence}, 171(5--6):255--285, 2007.

\bibitem[HHR09]{hem-hem-rot:j:hybrid}
E.~Hemaspaandra, L.~Hemaspaandra, and J.~Rothe.
\newblock Hybrid elections broaden complexity-theoretic resistance to control.
\newblock {\em Mathematical Logic Quarterly}, 55(4):397--424, 2009.

\bibitem[HP01]{hae-puk:j:electoral-writings-ramon-llull}
G.~H{\"{a}}gele and F.~Pukelsheim.
\newblock The electoral writings of {Ramon} {Llull}.
\newblock {\em Studia Lulliana}, 41(97):3--38, 2001.

\bibitem[Kem59]{kem:j:no-numbers}
J.~Kemeny.
\newblock Mathematics without numbers.
\newblock {\em Daedalus}, 88:577--591, 1959.

\bibitem[{Len}83]{len:j:integer-fixed}
H.~{Lenstra, Jr.}
\newblock Integer programming with a fixed number of variables.
\newblock {\em Mathematics of Operations Research}, 8(4):538--548, 1983.

\bibitem[LFZL09]{fen-liu-lua-zhu:j:parameterized-control}
H.~Liu, H.~Feng, D.~Zhu, and J.~Luan.
\newblock Parameterized computational complexity of control problems in voting
  systems.
\newblock {\em Theoretical Computer Science}, 410(27--29):2746--2753, 2009.

\bibitem[LZ10]{liu-zhu:j:maximin}
H.~Liu and D.~Zhu.
\newblock Parameterized complexity of control problems in maximin election.
\newblock {\em Information Processing Letters}, 110(10):383--388, 2010.

\bibitem[Men10]{men:t:range-voting}
C.~Menton.
\newblock Normalized range voting broadly resists control.
\newblock Technical Report arXiv:1005.5698~[cs.GT], arXiv.org, May 2010.

\bibitem[ML06]{mcl-lor:t:papacy}
I.~McLean and H.~Lorrey.
\newblock Voting in the medieval papacy and religious orders.
\newblock Report 2006-W12, Nuffield College Working Papers in Politics, Oxford,
  Great Britain, September 2006.

\bibitem[MLCM10]{che-mau-mon-lan:c:possible-winners-adding}
N.~Maudet, J.~Lang, Y.~Chevaleyre, and J.~Monnot.
\newblock Possible winners when new candidates are added: {T}he case of scoring
  rules.
\newblock In {\em Proceedings of the 24th AAAI Conference on Artificial
  Intelligence}, July 2010.
\newblock To appear.

\bibitem[MN08]{mes-nur:b:distance-realizability}
T.~Meskanen and H.~Nurmi.
\newblock Closeness counts in social choice.
\newblock In M.~Braham and F.~Steffen, editors, {\em Power, Freedom, and
  Voting}. Springer-Verlag, 2008.

\bibitem[MPRZ08]{mei-pro-ros-zoh:j:multiwinner}
R.~Meir, A.~Procaccia, J.~Rosenschein, and A.~Zohar.
\newblock The complexity of strategic behavior in multi-winner elections.
\newblock {\em Journal of Artificial Intelligence Research}, 33:149--178, 2008.

\bibitem[MPS08]{mcc-pri-sli:j:dodgson}
J.~{McCabe-Dansted}, G.~Pritchard, and A.~Slinko.
\newblock Approximability of {D}odgson's rule.
\newblock {\em Social Choice and Welfare}, 31(2):311--330, 2008.

\bibitem[MU95]{mcl-urk:b:polsci:classics}
I.~McLean and A.~Urken.
\newblock {\em Classics of Social Choice}.
\newblock University of Michigan Press, 1995.

\bibitem[Nie06]{nie:b:invitation-fpt}
R.~Niedermeier.
\newblock {\em Invitation to Fixed-Parameter Algorithms}.
\newblock Oxford University Press, 2006.

\bibitem[Pap94]{pap:b:complexity}
C.~Papadimitriou.
\newblock {\em Computational Complexity}.
\newblock Addison-Wesley, 1994.

\bibitem[Wal09]{wal:c:where-hard-veto}
T.~Walsh.
\newblock Where are the really hard manipulation problems? {T}he phase
  transition in manipulating the {V}eto rule.
\newblock In {\em Proceedings of the 21st International Joint Conference on
  Artificial Intelligence}, pages 324--329. AAAI Press, July 2009.

\bibitem[XC08a]{xia-con:c:generalized-scoring}
L.~Xia and V.~Conitzer.
\newblock Generalized scoring rules and the frequency of coalitional
  manipulability.
\newblock In {\em Proceedings of the 9th ACM Conference on Electronic
  Commerce}, pages 109--118. ACM Press, July 2008.

\bibitem[XC08b]{xia-con:c:frequently-manipulable}
L.~Xia and V.~Conitzer.
\newblock A sufficient condition for voting rules to be frequently manipulable.
\newblock In {\em Proceedings of the 9th ACM Conference on Electronic
  Commerce}, pages 99--108. ACM Press, July 2008.

\bibitem[XCPR09]{con-pro-ros-xia:c:unweighted-manipulation}
L.~Xia, V.~Conitzer, A.~Procaccia, and J.~Rosenschein.
\newblock Complexity of unweighted manipulation under some common voting rules.
\newblock In {\em Proceedings of the 21st International Joint Conference on
  Artificial Intelligence}, pages 348--353. AAAI Press, July 2009.

\bibitem[YL78]{lev-you:j:condorcet}
H.~Young and A.~Levenglick.
\newblock A consistent extension of {Condorcet}'s election principle.
\newblock {\em SIAM Journal on Applied Mathematics}, 35(2):285--300, 1978.

\bibitem[ZPR09]{pro-ros-zuc:j:borda}
M.~Zuckerman, A.~Procaccia, and J.~Rosenschein.
\newblock Algorithms for the coalitional manipulation problem.
\newblock {\em Artificial Intelligence}, 173(2):392--412, 2009.

\end{thebibliography}

\end{document}